\documentclass{article}
\usepackage[utf8]{inputenc}
\usepackage{fullpage}
\usepackage{amsfonts,amsmath,amssymb,amsthm}
\usepackage{cite,comment,enumerate,hyperref}
\usepackage{nicefrac}
\usepackage[colorinlistoftodos,textsize=small,color=red!25!white,obeyFinal]{todonotes}
\usepackage{csquotes}
\usepackage{wrapfig}

\newtheorem{theorem}{Theorem}
\newtheorem{corollary}[theorem]{Corollary}
\newtheorem{observation}{Observation}

\newtheorem{lemma}[theorem]{Lemma}
\newtheorem{claim}[theorem]{Claim}
\theoremstyle{definition}
\newtheorem{definition}{Definition}
\theoremstyle{remark}	

\usepackage{xcolor,xspace}

\usepackage{algorithm}
\usepackage{algorithmic}

\newcommand{\opt}{\mathsf{opt}}
\newcommand{\OPT}{\mathsf{opt}}

\usepackage{marginnote}
\usepackage{etoolbox}
\newtoggle{lmargin}
\newcommand{\alternatingtodo}[2][]{%
    \iftoggle{lmargin}%
    {%
        \todo[#1]{#2}%
        \togglefalse{lmargin}%
    }{%
        {%
            \let\marginpar\marginnote%
            \reversemarginpar%
            \todo[#1]{#2}%
        }%
        \toggletrue{lmargin}%
    }%
    \ignorespaces%
}

\title{Controlling Tail Risk in Online Ski-Rental}
\author{Michael Dinitz\thanks{Department of Computer Science, Johns Hopkins University, Baltimore, MD.  \texttt{mdinitz@cs.jhu.edu}.  Supported in part by NSF grants CCF-1909111 and CCF-2228995.  Work partially done while a Visiting Researcher at Google Research New York, NY.} \and 
Sungjin Im\thanks{Electrical Engineering and Computer Science, University of California, 5200 N. Lake Road, Merced CA 95344. \texttt{sim3@ucmerced.edu}. Supported in part by NSF grants CCF-1844939  
 and CCF-2121745.}\and
Thomas Lavastida\thanks{Jindal School of Management, University of Texas at Dallas, Richardson, TX.  \texttt{thomas.lavastida@utdallas.edu}.}\and 
Benjamin Moseley\thanks{Tepper School of Business, Carnegie Mellon University, Pittsburgh, PA. \texttt{moseleyb@andrew.cmu.edu}.  Work supported in part by a Google Research Award, an Infor Research Award, a Carnegie Bosch Junior Faculty Chair and NSF Grants CCF-2121744 and CCF-1845146. } \and 
Sergei Vassilvitskii\thanks{Google Research New York, NY.  \texttt{sergeiv@google.com}.}}
\date{}

\begin{document}

\maketitle
\thispagestyle{empty}

\begin{abstract}
The classical ski-rental problem admits a textbook 2-competitive deterministic algorithm, and a simple randomized algorithm that is $\nicefrac{e}{e-1}$-competitive in expectation. The randomized algorithm, while optimal in expectation, has a large variance in its performance: it has more than a 37\% chance of competitive ratio exceeding 2, and a $\Theta(1/n)$ chance of the competitive ratio exceeding $n$!

We ask what happens to the optimal solution if we insist that the {\em tail risk}, i.e. the chance of the competitive ratio exceeding a specific value is bounded by some constant $\delta$. We find that this additional modification significantly changes the structure of the optimal solution. The probability of purchasing skis on a given day becomes non-monotone, discontinuous, and arbitrarily large (for sufficiently small tail risk $\delta$ and large purchase cost $n$). 

\end{abstract}

\newpage
 \setcounter {page} {1} 
\section{Introduction}
\label{sec:intro}
Decision-making under uncertainty about the future is a central topic in algorithm design; online algorithms, studied through the metric of competitive analysis, have been successful in guaranteeing worst-case performance against adversarial inputs.  Arguably, the most basic online problem is the \emph{Ski Rental} problem, which captures a commonly faced sub-problem, usually known as ``rent or buy'': we need to decide whether to stay in the current state, paying some cost per time unit, or switch to another state, which is expensive but requires no further payment.  In the specific ski rental problem, every morning Alice must decide whether to rent skis for \$1 or buy them for \$$n$, in which case she never needs to rent them again. Her choice is non-obvious because she does not know the number of days, denoted as $x$, she is going to ski---the weather may become too warm, she may get injured, or may just get tired of the sport.

Folklore analysis says that committing to buying skis on the morning of day $n$ is (deterministically) optimal, as Alice never over-spends by more than a factor of $2$ (i.e., this approach has a competitive ratio bounded by $2$), no matter how many days she ends up skiing.  If Alice is willing to randomize, she can do even better---she can commit to buying the skis on day $i\in [n]$ with probability proportional to $\exp(i/n)$. This method gives the best possible $\nicefrac{e}{e-1} \approx 1.58$ competitive ratio in expectation~\cite{KarlinMRS88}.

The competitive ratio is for the worst-case action by the adversary (who decides on the number of skiing days), it only holds {\em in expectation} for Alice. An easy calculation shows that an adversary that ends the ski season on day $n/2$ ensures that Alice exceeds the competitive ratio of 2 with probability $(\sqrt{e} - 1)/(e-1) \approx 37\%$; for more details, see Appendix~\ref{app:worse-than-deterministic}.  Thus more than a third of the time, Alice is better off following the deterministic strategy. Furthermore, the competitive ratio is $\Omega(n)$ with probability $\Omega(1/ n)$\footnote{This can occur when $x = \Theta(1)$, yet Alice ends up buying skis.}, illustrating that the ``best'' randomized algorithm has a considerable chance of returning a solution significantly worse than the deterministic alternative.  Importantly, in the online setting decisions are irrevocable, so results in expectation do not immediately lead to high-probability bounds. This is in contrast to an offline algorithm where bounding the approximation ratio in expectation often leads to giving the same bound with high probability by running the algorithm a logarithmic number of times independently and taking the best solution.

What if one desires upper-case bounds on the chance the randomized algorithm is worse than the deterministic algorithm? A natural direction is to find the optimal algorithm and study its competitive ratio as a function of the probability of the competitive ratio exceeding 2. We have the two endpoints---the deterministic algorithm with a ratio of $2$, and the randomized algorithm with an {\em expected} ratio of $\nicefrac{e}{e-1}$ and a probability of $\nicefrac{(\sqrt{e} - 1)}{(e-1)}$ of exceeding $2$. The question we study in this work is what happens in between?  What is possible if we put nontrivial constraints on the tail performance, and what is the structure of these optimal solutions?

\begin{figure*}[t!]
\includegraphics[width=0.99\textwidth]{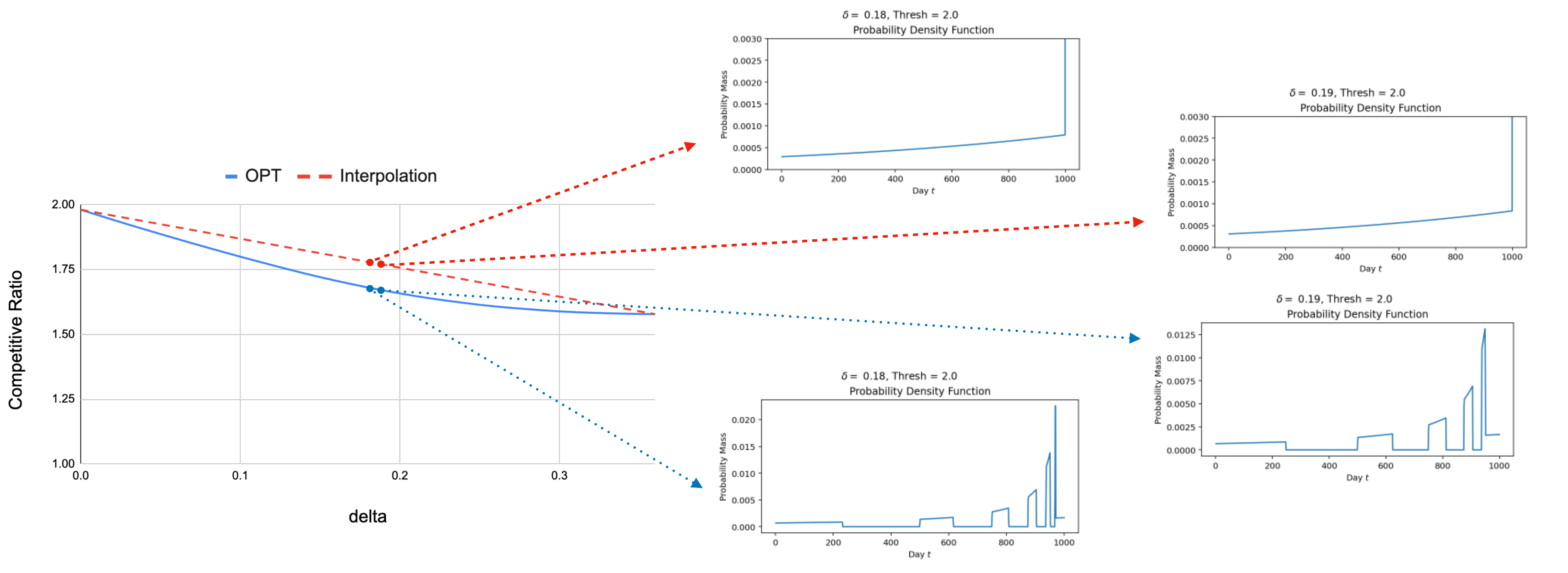}
\caption{Figure of the competitive ratio as a function of $\delta$ for both the optimal strategy and the interpolation strategy. For fixed values of $\delta = \{0.18, 0.19\}$ we show the probability of purchasing skis on day $t \in [1000]$ for the two algorithms. Observe that the interpolation algorithm (red dashed line) puts a  weight of $1 -  \delta/\delta^* \approx \nicefrac{1}{2}$ on the last day in both cases. The optimal algorithm distributes the weight across the days, but does so in a non-continuous manner. For instance, observe that the figure for $\delta = 0.18$ has six non-zero intervals, whereas $\delta = 0.19$ has five non-zero intervals.   }
\label{fig:ratios}
\end{figure*}

\iffalse
\begin{figure*}[t!]

    \includegraphics[width=0.33\textwidth]{del45.png}
    \includegraphics[width=0.33\textwidth]{del25.png}
    \includegraphics[width=0.33\textwidth]{del15.png}
%    \includegraphics[width=2in]{del35.png}
\caption{The probability of purchasing skis on day $t$ of $1,000$ used by the optimal randomized algorithm. (a) no constraints on tail risk of competitive ratio above $\gamma = 2$,  (b) $\delta = 0.25$ and (c) $\delta$ = 0.15. Panes (b) and (c) clearly show the drops to 0 probability of purchase and the exponentially growing rate of purchase. The corresponding competitive ratios are: (a) $\nicefrac{e}{e-1} \approx 1.58$, (b) $\approx 1.62$ and (c) $\approx 1.73$.}
\label{fig:plots}
\end{figure*}
\fi

To formalize the question we study, let a $(\gamma, \delta)$-tail constraint denote the restriction that the probability (over the choices made by the algorithm) that the worst-case competitive ratio exceeds $\gamma$ is at most $\delta$. Let $\delta^* = \nicefrac{(\sqrt{e} -1)}{(e-1)}$. For the ski rental problem, the deterministic algorithm is optimal and satisfies a $(2, 0)$-tail constraint, whereas the randomized algorithm optimizes the competitive ratio among all of the algorithms satisfying a $(2, \delta^*)$-tail constraint, or more generally, a $(2, \delta)$-tail constraint for any $\delta \geq \delta^*$.  

Given a collection of arbitrary $(\gamma, \delta)$ constraints, or even a single nontrivial tail constraint, what is the optimal algorithm? One may suspect that the solution is well-structured with behavior that is similar to the classical deterministic and randomized solutions.  For example, suppose that we are given a single $(2, \delta)$-tail constraint with $\delta < \delta^*$.  It it not hard to see that one way of achieving this is to interpolate between the two classical solutions: with probability $\delta/\delta^*$ we run the classical randomized algorithm, and with probability $1-\delta/\delta^*$ we run the classical deterministic algorithm.  This clearly satisfies the tail constraint, is a very simple algorithm, and inherits a number of nice properties (monotonically increasing probabilities which increase only exponentially, a single discontinuity at day $n$, non-zero probability on every day, etc.).  Is this interpolation optimal?  If not, does the optimal solution at least have these nice properties, or look ``simple''?  

\subsection{Our Contribution}
The answer to these questions is an emphatic \emph{no}. Not only is the the optimal solution \emph{not} an interpolation between the classical algorithms, its structure is wildly different from any previously considered ski rental algorithm and does not satisfy any of the ``nice'' properties mentioned earlier.  In particular, we show that the optimum solution has the following surprising properties:

\begin{itemize}
    \item \textbf{Non-monotonicity.} The purchase probability does not grow monotonically in time $t$.
    \item \textbf{Arbitrarily many discontinuities.} Even for a single $(\gamma, \delta)$ tail constraint, the purchase probability becomes zero and then positive $\Omega(1/ \delta)$ times for 
    any sufficiently large $n$ compared to any fixed sufficiently small $\delta$.    
    \item \textbf{Arbitrarily fast-growing purchase probability.} Again, even for a single $(\gamma, \delta)$ tail constraint, when the probability becomes positive, it grows much faster than before---the exponent doubles each time.  This results in continuous sections of the purchase distribution with arbitrarily fast growth. 
\end{itemize}

 To illustrate these points, we show the interpolation and the optimal competitive ratios as a function of $\delta$ in Figure \ref{fig:ratios}. To add, we show the purchase distributions for two nearby points (namely $\delta = 0.18$ and $\delta = 0.19$) for both cases. Note how the solution, while structured enough to reason about, is not in any way ``smooth''.  For example, even for these very close $\delta$, the number of nonzero regions is distinct---there is a ``discontinuity'' in how the optimal solution behaves as $\delta$ changes from $0.19$ to $0.18$.  This is despite the fact that the achieved competitive ratio \emph{is} smoothly changing with $\delta$, as is apparent from Figure~\ref{fig:ratios}.

Proving these properties of the optimal solution is our main technical contribution.  The formal statements are quite complex, but can be found in Section~\ref{sec:general-gamma}, particularly, Lemma~\ref{lem:explicit-general-gamma} and Corollary~\ref{cor:big-exponentials-general-gamma}.
As a side effect of our techniques we can also design an algorithm to actually compute the optimal solution (see Section~\ref{sec:algorithm}). 

We would like to emphasize that ski rental is traditionally considered to be an extraordinarily \emph{simple} setting for online algorithms.  The fact that adding a single tail constraint (which as discussed is something we naturally want in an online setting where we cannot run multiple times) results in such a complex and nonobvious solution structure is, in our opinion, extremely surprising.  We hope that this opens up a new set of questions on tail bounds in online algorithms. 

\subsubsection{Technical Overview}

In order to introduce our techniques, we first need to introduce some notation.  Consider the following  characterization of any randomized algorithm for ski rental. Prior to skiing on the first day, Alice flips a coin and commits to buying skis on the morning of day $i$ with probability $f_i$. Note the $f_i$'s form a distribution, i.e., $\sum_i f_i = 1$. We call $f = \{f_i\}_{i\in [n]}$, the {\em purchase distribution}. 

The purchase distribution, along with the adversary's choice of the last skiing day $x$, induces the competitive ratio, which itself is a random variable. We denote the {\em competitiveness} function $\alpha_f$, where $\alpha_f(x)$ is the {\em expected} competitive ratio when Alice stops skiing at the end of day $x$ and chooses her purchase day from the distribution $f$. 

Let $f^*$ be the randomized algorithm minimizing the expected competitive ratio, i.e., minimizing $\max_x \alpha_{f^*}(x)$.  It is known that this algorithm sets $f^*_i \propto \exp(i)$. The optimal choice of $f^*$ balances two competing objectives: buying early is good if the number of skiing days exceeds $n$, on the other hand, buying late is good if the adversary's choice for the number of skiing days is small. The balance is done in the worst-case over all adversary's choices, but in expectation over the random coins flipped by the algorithm. In the optimal solution the balance has the effect that the competitiveness function $\alpha_{f^*}(x) = \nicefrac{e}{e-1}$, that is, it is constant over all $x \in [n]$.  Thus, no matter what day the adversary picks as the last ski day, Alice will have the exact same expected competitive ratio of $\nicefrac{e}{e-1}$.

\paragraph{Optimal Solution Characterization}

In order to understand the structure of the optimal purchase distribution in the presence of tail bounds, we first need to give a characterization of this optimal distribution.  It turns out that this characterization will naturally lead to an efficient algorithm to \emph{construct} the optimal purchase distribution under any combination of tail constraints, although there are some technical complications which have to be overcome along the way (see Section~\ref{sec:algorithm} for details).  While our characterization and algorithm holds for arbitrary combinations of tail constraints, for simplicity we will now assume that there is only a single $(\gamma, \delta)$-tail constraint.

For every day $x$, if we condition on the adversary choosing $x$ to be the final day of skiing, then there are essentially two constraints which the optimal purchase distribution $f$ must obey:  i) \emph{competitiveness constraint}: the expected competitive ratio using $f$ must be at most the optimal competitive ratio $\opt$, i.e., $\alpha_f(x) \leq \opt$; and ii) \emph{tail constraint}: the probability when we choose a purchase day from $f$ that we achieve competitive ratio worse than $\gamma$ must be at most $\delta$ (for a formal version of this in terms of ``bad intervals'', see Section~\ref{sec:general-tail}). 

These constraints are both inequalities, but our main characterization theorem (see Theorem~\ref{thm:opt-tight-general} in Section~\ref{sec:characterization}) is that for every $x$, \emph{at least one of these is tight}.  That is, for every day $x$, when we condition on the adversary picking day $x$ as the final day, \emph{either} the competitive ratio of $f$ is equal to $\opt$ \emph{or} the probability of $f$ achieving competitive ratio worse than $\gamma$ is equal to $\delta$.  

To gain some intuition for this, recall that if we do not have tail constraints, then in the optimal purchase distribution the competitive ratio is exactly the same, ($\nicefrac{e}{e-1}$), for every time $x$ which the adversary might choose.  This is no longer true in the presence of tail constraints, but if there is some day when the competitive ratio is strictly less than the optimal competitive ratio, then it can only be because the tail bound is tight on that day.  

This structure theorem leads to solutions with surprisingly complex structure.  But it also means that this complex structure is purely a function of which constraint is tight at which days.  This is also (at a very high level) how our algorithm works: we guess $\opt$ (which is non-obvious, see Appendix~\ref{sec:guessing_OPT}), and then can iteratively set the probability for each day to be whatever makes one of the two constraints tight.  To see this theorem ``in action'', consider Appendix \ref{app:figure}.  
The days in which the expected competitive ratio dips below $\opt$ are precisely the days where the tail constraint is tight, and these transitions also obviously correspond to transitions in the purchase distribution.

Notably, while this characterization allows us to find the optimal purchase distribution algorithmically, the exact value of the purchase distribution may not have a closed form with elementary functions. We give an example in Appendix \ref{app:lambert} where the exact solution involves the Lambert $W$ function.

\paragraph{Single Tail Constraint}  With the characterization of the optimal solution in place, we  consider  the solution in the presence of a single $(\gamma, \delta)$ tail constraint.  Here the pair of invariants we described above leads to a solution with a non-trivial shape, which we illustrate in the bottom panes of Figure \ref{fig:ratios}.
Observe that the purchase distribution, $f$,  oscillates between periods of $0$ probability of buying, together with periods of ever higher probability of buying. 

Although we do not compute the exact closed form of the competitive ratio due to the difficulty mentioned above, we characterize the optimum purchase probability distribution quite precisely when $t$ is small---including exactly the times when the probability distribution drops to 0 and the exact exponentially growing exponents. 

As the optimum solution for the discrete version of ski-rental problem is technically complex, to make our presentation transparent, here we only present our result for the continuous setting. In the continuous setting, by scaling we can assume without loss of generality that Alice can buy skis for \$1 and rent for $\$dt$ per $dt$ time; she will ski for $x$ time where $x$ is a fractional value in $(0, 1]$. Thus, the randomized algorithm can be described as a purchase probability distribution $f(t)$ over $(0, 1]$; see Section~\ref{sec:prelim} for more details. The following theorem reveals the optimum solution's structure.

\begin{theorem}
    \label{thm:single-tail-intro}    \emph{(Corollary~\ref{cor:continuous-general-gamma})}
    Let $\gamma \geq 2$ be an integer. 
    Let $\{P_j\}_{j \geq 0}$ be a partition of $(0, 1 / (\gamma - 1)]$, defined as $P_0 = (0, \frac{1}{\gamma}]$ and  $P_j = (\frac{1}{\gamma-1} \cdot (1 - \frac{1}{\gamma^j}), \frac{1}{\gamma-1} \cdot (1 - \frac{1}{\gamma^{j+1}})]$ for all $j \geq 1$. Let $\ell_j$ be the start point of $P_j$.
    In the continuous version of the ski rental problem with a single tail constraint $(\gamma, \delta)$, the optimum solution $f(t)$ has the following structure in every $P_j$ such that $j 
 < 1 / (2\gamma\delta) - 1$:
    \begin{itemize}
        \item $f(t) = c_j \cdot e^{\gamma^j t}$ for $t$ in the interval $(\ell_j, \ell_j + \hat t_j]$, and
        \item $f(t) = 0$ for all $t \in (\ell_j + \hat t_j, \ell_{j+1}]$,
    \end{itemize}
    where $c_j$ and $\hat t_j <  \ell_{j+1} - \ell_j$ are some constants depending on $\delta$ and $j$.
\end{theorem}

Let's assume $\gamma = 2$ to illustrate the theorem. 
The theorem shows that with a single tail constraint $(2, \delta)$, the optimum solution has certain recurring structures over time intervals, $P_0, P_1, \ldots $, which are of exponentially decreasing lengths by a factor of 2. Here, note that we can only show the structure for early intervals. For example, if $\delta < 1/ 12$, we can analytically show the recurring structure over $P_0 = (0, 1/2], P_1 = (1/2, 3/4]$ and $P_2 = (3/4, 7/8]$. In such an interval $P_j$, $f(t)$ initially grows in proportion to $e^{2^j t}$ and then drops to 0 and remains 0 until the next interval $P_{j+1}$ starts. 

We sketch the analysis of this result. We begin with a key observation that the probability mass of $f$ on each $P_j$ must be at most $\delta$, which follows from taking a close look at the tail constraint. More precisely, the tail constraint can be shown to be equivalent to the probability mass on each interval $I(x)$ ending at $x$ being at most $\delta$. It turns out that $P_0, P_1, \ldots $ are such intervals. 

We show that in beginning of $P_0$, the competitiveness constraint must be tight, thus $f$ grows in proportion to $e^t$ and then at some point to respect the tail constraint it drops to 0. Afterwards, we show that the tail constraint must be tight. Since $f$ must restart collecting probability mass from the beginning of each $P_j$, which is shorter than the previous interval $P_{j-1}$, it has to accelerate the collection.

The actual analysis is based on a careful induction argument, but this sketches the high-level proof idea.
We note that in the discrete version we can show that 
the number of times $f$ drops to 0 is at least $\Omega(1/\delta)$ (for large enough $n$).  See Section~\ref{sec:general-gamma} for more details.

\paragraph{Single Pure Tail Constraint}
 As a side effect of our techniques, we are actually able to give a precise and explicit description of the optimal solution when there is a single tail constraint of the form $(\gamma, 0)$, i.e., the competitive ratio is never allowed to exceed $\gamma$. We call these {\em pure} tail constraints.  A simple calculation shows that such a constraint implies that $f_i = 0$ for $i < n/(\gamma - 1)$: we must always rent for the first $n/(\gamma - 1)$ days. Since this is the only tail constraint, the remaining question is how to allocate the probability mass such that maximum of the competitiveness function on the remaining interval is minimized. In this case we can precisely show the value of the optimal competitive ratio and the shape of the optimal solution as follows (see the theorems in Appendix~\ref{sec:pure} for more precise statements).

\begin{theorem} 
    In the ski-rental problem with tail constraint $(\gamma, 0)$, the optimum solution is the following (assuming $\gamma - 1$ divides $n-1$): 
    \begin{align*}
        f_t = \begin{cases} 0 & t < \frac{n-1}{\gamma-1} \\ 
        \frac{\lambda-1}{\gamma}-1 & t = \frac{n-1}{\gamma-1} \\
        \frac{\gamma(\lambda-1)}{(n-1)(\gamma-1)} \cdot \left( 1 + \frac{1}{n-1}\right)^{t - t_1 - 1} & t > \frac{n-1}{\gamma-1} 
        \end{cases}
    \end{align*}
    where the optimum competitive ratio $\lambda$ tends to  $     1 + \frac{\gamma-1}{1 + \gamma \left(e^{1 - \frac{1}{\gamma-1}} - 1\right) } $ as $n \rightarrow \infty$.
\end{theorem}

Note that when $\gamma = 2$ we get that $\lim_{n \rightarrow \infty} \lambda = 2$, i.e., we recover the classical deterministic bound, and when $\gamma \rightarrow \infty$ we get that $\lim_{n \rightarrow \infty} \lambda = 1 + \frac{1}{e-1} = \frac{e}{e-1}$, i.e., we recover the classical randomized bound.

\subsection{Related Work} \label{sec:related}

\paragraph{Ski-Rental and its Variants:} The classical ski-rental problem captures the fundamental ``rent or buy'' dilemma that exists at the heart of many online problems.  The deterministic 2-competitive break-even strategy was first analyzed by Karlin et al.~\cite{KarlinMRS88} as a special case of Snoopy Caching and later the optimal $e/(e-1)$-competitive randomized strategy was described in \cite{KarlinMMO94}. We note that the online primal-dual approach was used to give an optimal $e/(e-1)$-competitive algorithm \cite{buchbinder2009design}. In separate work, Karlin et al.~
 \cite{KarlinKR03} gave optimal $e/(e-1)$-competitive algorithms for dynamic TCP acknowledgment and the Bahncard problem by exploiting a connection to the classic ski-rental problem. Here, we expand the classical ski-rental problem to include the tail of the competitive ratio as a new metric in addition to the expected competitive ratio.  Even with this small change, we find that the optimal algorithm exhibits interesting behaviour that is not observed in the standard setting.

There are several lines of work concerned with generalizing and analyzing new variants of the classical ski-rental problem.  Motivated by applications in cloud cost optimization, Khanafer et al.~\cite{KhanaferMPK2013} considered a variant of ski-rental where the number of days is randomly chosen from a distribution with known first or second moments (but otherwise unspecified). Wang et al.~\cite{WuBY21} consider a variant of ski-rental with multiple commodities that can be rented, purchased individually, or purchased as a group.  Recently there has been significant interest in incorporating predictions derived from machine learning into online algorithms ~\cite{MVSurvey} which has resulted in a sequence of work applying this to the ski-rental problem \cite{PurohitSK18,Banerjee20,WangLW20,WeiZ20,AntoniadisCEPS21,Zeynali0HW21,ShahR21}.

\paragraph{Online Algorithms Beyond Expectation:} It is standard to analyze the expected performance of algorithms in several contexts, such as minimizing the competitive ratio of an online algorithm or the regret of a policy in a multi-arm bandit setting.  As we point out for the ski-rental problem, there may be cases where tailoring an algorithm to the mean may not be sufficient and other metrics such as the tail may be important.  As another example, in multi-arm bandits achieving low regret on average necessitates a certain amount of exploration which can increase the variability in the attained reward.  Wu et al.~\cite{WuSLS16} consider the conservative bandits setting in which the goal is to minimize the cumulative regret of the policy while constraining the total reward earned at each time to be above some baseline level.  Another example is the survival bandit problem~\cite{PerottoBSV19, riou2023survival} in which rewards can be both positive and negative and the objective is to maximize the total reward while keeping the ``risk of ruin''---the probability of the current total reward going below a fixed threshold---small.

\section{Preliminaries}
\label{sec:prelim}

In the (discrete version of) ski rental problem, also a.k.a. rent-or-buy, every morning Alice must decide whether to rent skis for \$1 or buy them for \$$n$, in which case she never needs to rent them again. The number of days (times), denoted as $x$, she will ski is unknown a priori, and it is revealed to Alice only at the end of day $x$. 

Although Alice is allowed to rent beyond day $n$, a moment's thought shows that she wouldn't benefit it. Thus, without loss of generality, we can describe any algorithm by its purchase distribution $f = \{f_t\}_{t \in [n]}$. 
  For $x,t \in [n]$, let $\alpha(t,x)$ denote the competitive ratio of the algorithm when the adversary chooses for $x$ to be the last skiing day and Alice buys skis on the morning of day $t$.  We have that, 
\[
\alpha(t,x) = \begin{cases}
    1 & \text{if } t > x \\
    \frac{n+t-1}{x} & \text{if } t \leq x
\end{cases}
\]

Given a distribution $f$ over $[n]$, let $\alpha_f(x)$ be the expectation of $\alpha(t,x)$ where $t$ is drawn from $f$.  So we have that
\begin{align}
\alpha_f(x) &= \sum_{t \in [n]} \alpha(t,x) f_t = \sum_{t \leq x} \frac{n+t-1}{x} f_t + \sum_{t > x} f_t %\notag \\&
= \sum_{t \leq x} \frac{n+t-1}{x} f_t + 1 - \sum_{t \leq x} f_t \label{eq:alpha}
\end{align}

The following observations are immediate from the definition.

\begin{observation}
    \label{obs:bad-interval}
    Given a purchase distribution $f$, 
    $\alpha_f(t)$ only depends on $f_1, f_2, \ldots, f_t$. If we move the probability mass from a later time to an earlier time, $\alpha_f(n)$ decreases.     
\end{observation}

It is well known that the  deterministic algorithm that buys on day $n$, i.e., $f_1 = \ldots = f_{n-1} = 0$ and $f_n = 1$, is $2 - 1/ n$-competitive and it is the best possible competitive ratio achievable for any deterministic algorithms. Also, when $f_t = (1 / n) (1 + 1/n)^{t -1} / ((1 + 1/ n)^n - 1)$ we recover the celebrated randomized algorithm whose competitive ratio is $(1 + 1/ n)^n  / ((1 + 1/ n)^n - 1)$, which tends to $e / (e - 1)$ as $n \rightarrow \infty$. An easy calculation shows $\alpha_f(x)$ is a constant which is exactly the claimed competitive ratio and in fact keeping $\alpha_f(x)$ constant for all $x$ makes the purchase probability grows exponentially in $t$.

\paragraph{Continuous Case.}
The ski rental problem is often discussed in the continuous setting as it exhibits a cleaner representation of the optimum solution; for example, see~\cite{KarlinKR03,deb-ski-lecture}.
More precisely, after scaling, we can assume wlog that  Alice can ski for $x$ amount of time where $x \in (0, 1]$ and the rental price for $dt$ amount of time is $dt$ and the purchase price is $1$. Then, the optimum randomized algorithm has pdf $f(t) = e^t / (e - 1)$ and the competitive ratio is exactly $e / (e - 1)$. In our problem with tail constraints, we will mostly consider the discrete version as it seems to resist a simple closed form for the competitive ratio unlike the problem  without tail constraints. 

\paragraph{A Useful Lemma.}
The following lemma (proof in Appendix~\ref{app:omitted-prelims}) will prove to be useful in a few different places in our analysis.  To interpret it, first note that thanks to Equation~\eqref{eq:alpha}, the competitive ratio at some time $x$ can be calculated using only the probabilities for times $t \leq x$.  So the following lemma lets us say that if we have built ``part'' of the distribution, then we can extend it so that the competitive ratio stays constant by increasing each successive probability by a multiplicative $\left(1+ \frac{1}{n-1}\right)$ factor.  In other words, an appropriate exponential function keeps the competitive ratio constant.  This explains, for example, the appearance of an exponential function in Figure~\ref{fig:constraints}(a) at the same times the competitive ratio is flat in Figure~\ref{fig:constraints}(b).

\begin{lemma} \label{lem:discrete-exponential}
    Let $1 < a < n$, let $f : [a-1] \rightarrow [0,1]$ such that $\sum_{t=1}^{a-1} f_t < 1$, and let $a \leq x' \leq n$.  Then $\alpha_f(x) = \alpha_f(a-1)$ for all $a \leq x \leq x'$ if and only if a) $f_x = \left(1+\frac{1}{n-1}\right)^{x-a} f_a$ for all $a \leq x \leq x'$, where $f_a = \frac{1}{(a-1)(n-1)} \sum_{t=1}^{a-1} (n+t-1)f_t$, and b) $\sum_{t=1}^{x'} f_t \leq 1$. 
\end{lemma}

\section{Characterizing the Optimal Solution} \label{sec:general-tail}
In this section, we consider the general case where there can be multiple tail constraints, each with an arbitrary threshold at least $2$.  Our goal is to prove a characterization of the optimal purchase distribution that, while not necessarily allowing us to write  it explicitly, will enable us to reason about its properties.  Moreover, a side effect of our characterization will be an efficient algorithm to actually construct the optimal purchase distribution.  

Suppose that we are given a collection $\{(\gamma_i, \delta_i)\}_{i \in [k]}$ where $\gamma_i \geq 2-1/n$ and $0 \leq \delta_i < 1$ for all $i \in [k]$\footnote{Under this assumption we have that $f_n$ does not affect any tail constraints, see Definition~\ref{def:bad_interval_def} and Observation~\ref{obs:no-bad-n}.  We focus on the case that $\gamma_i \geq 2-1/n$ since 
 $2 - 1/n$ is the competitive ratio that is achieved by a deterministic algorithm, but note that smaller values of $\gamma_i$ can be handled so long as $\delta_i$ is large enough for there to exist a feasible solution.}  The goal is to find the randomized algorithm with minimum expected competitive ratio (worst case over times the adversary might choose) subject to the requirements that for all $i \in [k]$, the probability that our algorithm has a competitive ratio larger than $\gamma_i$ is at most $\delta_i$, where again this is the worst-case over times the adversary might choose. The optimal solution's properties we will characterize for general tail constraints will be useful for discovering some surprising structural properties of the optimal solution later for
a single tail constraint.

\begin{definition} \label{def:bad_interval_def}
The \emph{bad interval} for threshold $\gamma$ and time $x$, which we denote by $I_{\gamma}(x)$, consists of all $t \in [x]$ such that if we buy at time $t$ and the adversary chooses time $x$, then our competitive ratio is larger than $\gamma$.  By definition, 
\[
I_{\gamma}(x) = \{t \in [n] : \alpha(t,x) > \gamma\} = \left\{ t \in [x] : \frac{n+t-1}{x} > \gamma \right\} = \left\{ t : \max(0, \gamma x - n +1) < t \leq x\right\}
\]
\end{definition}

\begin{observation}
    \label{obs:no-bad-n}
    For all $\gamma \geq 2 - 1/ n$ and all $x$, it is the case that $n \not\in I_{\gamma}(x)$.
\end{observation}

We say that a purchase distribution $f$ is feasible if it satisfies all tail constraints. The following lemma
shows the solution is feasible if and only if it satisfies a collection of packing constraints pertaining to the bad intervals.
The proof, as well as all the other missing proofs, are deferred to Appendix~\ref{app:omitted-general-tail}.

\begin{lemma} \label{lem:bad-interval}
A distribution $f$ is feasible if and only if $\sum_{t \in I_{\gamma_i}(x)} f_t \leq \delta_i$ for all $i \in [k]$ and for all $x \in [n]$.
\end{lemma}
%%%%%%%%%%%%%%%%%%% Appendix
\iffalse
\begin{proof}
    Let $f$ be a feasible distribution.  Then for all $i \in [k]$ and $x \in [n]$, we know that $\Pr_{t \sim f}[\alpha(t,x) > \gamma_i] \leq \delta_i$ since $f$ is feasible.  Hence $\sum_{t \in I_{\gamma_i}(x)} f_t \leq \delta$ by the definition of $I_{\gamma_i}(x)$ as required.  

    Now suppose that $\sum_{t \in I_{\gamma_i}(x)} f_t \leq \delta_i$ for all $i \in [k]$ and for all $x \in [n]$.  Then by the definition of $I_{\gamma_i}(x)$, we have that
    \[
    \Pr_{t \sim f}[\alpha(t,x) > \gamma_i] = \sum_{t \in I_{\gamma_i}(x)} f_t \leq \delta_i
    \]
    as required. 
\end{proof}
\fi

\subsection{Main Characterization Theorem} \label{sec:characterization}
Let $f^*$ be the optimal solution.  By Lemma~\ref{lem:bad-interval}, this means that $f^*$ is the distribution $f$ minimizing $\max_{x \in [n]} \alpha_{f}(x)$ subject to $\sum_{t \in I_{\gamma_i}(x)} f_t \leq \delta_i$ for all $i \in [k]$ and for all $x \in [n]$.  Let $\opt = \max_{x \in [n]} \alpha_{f^*}(x)$.

We now prove some useful properties about $f^*$ and $\opt$.  First we show that $\opt < 2-1/n$ whenever the set of tail bounds is $\{(\gamma_i, \delta_i)\}_{i\in [k]}$ with $\gamma_i \geq 2 - 1/ n$. Since the $2 - 1/ n$-competitive deterministic algorithm satisfies all tails constraints, it must be the case that $\opt \leq 2 - 1/ n$. Intuitively, randomization should yield a better competitive ratio and the following lemma formally proves it. 

\begin{lemma} \label{lem:better_than_deterministic_soln}
    If $\gamma_i \geq 2-1/n$ for all $i\in[k]$ then $\OPT < 2-1/n$.
\end{lemma}
%%%%%%%%%%%%  Appendix
\iffalse
\begin{proof}
    Let $f$ be the distribution which puts mass 1 on day $n$ and mass 0 elsewhere.  This distribution satisfies all tail constraints since $\Pr_{t \sim f} [\alpha(t,x) > \gamma_i] \leq \Pr_{t \sim f} [\alpha(t,x) > 2-1/n] = 0 \leq \delta_i$.  Note that $\alpha_{f}(x) \leq 2-1/n$ for all $x \in [n]$.  We construct a new distribution $f'$ by moving $\epsilon>0$ mass from $f_n$ to $f_{n-1}$ and show that $\alpha_{f'}(x) < 2-1/n$.  Let $f'$ be the distribution given by
    \[
    f'_t =
    \begin{cases}
        1-\epsilon \quad & \text{if } t=n \\
        \epsilon \quad & \text{if } t=n-1 \\
        0 \quad & \text{otherwise}.
    \end{cases}
    \]
    To check the tail constraints we note that the only non-empty bad interval that matters corresponds to $x = n-1$ since $I_{\gamma_i}(n) = \emptyset$ for $\gamma_i \geq 2-1/n$ as we assume (Observation 
    \ref{obs:no-bad-n}).    
    Thus as long as we choose $\epsilon$ so that $\sum_{t \in I_{\gamma_i}(n-1)} f'_t = \epsilon \leq \delta_i$ for all $i\in [k]$ we have a feasible distribution.  Next, it can be shown that
    \[
    \alpha_{f'}(x) = \begin{cases}
        1+ \epsilon \quad & \text{if } x=n-1 \\
        2 + \frac{1-\epsilon}{n} \quad & \text{if } x=n \\
        1 \quad & \text{otherwise}.
    \end{cases}
    \]
    Thus if we take $\epsilon \in (0,\min_i \delta_i)$, we have $\alpha_{f'}(x) < 2-1/n$ for all $x$, which certifies that $\OPT < 2-1/n$.
\end{proof}
\fi

By definition of $\opt$, it is immediate that $\alpha_{f^*}(x) \leq \opt$ for any adversarial choice of $x$. The following lemma shows that the  competitiveness function $\alpha$ is maximized on the last day $n$ for the optimum distribution $f^*$. The proof easily follows from Observations~\ref{obs:bad-interval} and \ref{obs:no-bad-n}: If the lemma were false, 
 we can move the probability mass from an earlier time to time $n$, which keeps the solution feasible while increasing $\alpha_{f^*}(n)$.

\begin{lemma} \label{lem:CR-n}
    Let $f^*$ be an optimal solution.  Then $\alpha_{f^*}(n) = \OPT$.
\end{lemma}

Now we can prove the main structural theorem that shows that for any time $x$ the optimum solution must have the maximum competitiveness function value or the bad time interval packing constraint at the time must be tight. We show the theorem by showing that if the theorem were false at time $t_1$, then we can move a probability mass to $t_1$ from a later time, thereby improving the competitive ratio. 

\begin{theorem} \label{thm:opt-tight-general}
    Let $f^*$ be an optimal solution.  Then for every $x \in [n]$, at least one of the following is true:    
    \begin{itemize}
        \item (competitiveness constraint) $\alpha_{f^*}(x) = \OPT$, or
        \item (tail constraint) $\sum_{t \in I_{\gamma_i}(x)} f^*_t = \delta_i$ for some $i \in [k]$
    \end{itemize}
\end{theorem}
\begin{proof}
    Suppose for contradiction that this is false.  Let $t_1$ be the smallest value for which both conditions are false, and observe that Lemma~\ref{lem:CR-n} implies that $t_1 < n$.  Let $t_2 > t_1$ be the smallest value larger than $t_1$ such that at least one of the two conditions holds, and again observe that Lemma~\ref{lem:CR-n} implies that such a $t_2$ must exist.  It it also easy to see that $f^*_{t_2} > 0$, since otherwise at least one of the conditions must hold at $t_2-1$. 
    More formally, suppose $f^*_{t_2} = 0$ for the sake of contradiction. By definition of $t_2$, we have $\alpha_{f^*}(t_2 -1) < \opt$. It is easy to verify 
    $\alpha_{f^*}(t_2 -1) > \alpha_{f^*}(t_2)$; thus we have $\alpha_{f^*}(t_2) < \opt$. Again by definition of $t_2$, we know $\sum_{t \in I_{\gamma_i}(t_2-1)} f^*_t < \delta_i$ for any $i \in [n]$. Using the facts that  $I_{\gamma_i}(t_2) \subseteq  I_{\gamma_i}(t_2 -1) \cup \{t_2\}$ and $f^*_{t_2} = 0$, we have 
    $\sum_{t \in I_{\gamma_i}(t_2)} f^*_t \leq \sum_{t \in I_{\gamma_i}(t_2-1)} f^*_t < \delta_i$. This implies none of the conditions hold true at time $t_2$, which is a contradiction.     
    
    Therefore, we can define the distribution that we get by shifting some very small $\epsilon > 0$ mass from $t_2$ to $t_1$, i.e., $f_t = f^*_t + \epsilon$  if  $t = t_1$; $f_t = f^*_t - \epsilon$   if  $t = t_2$; \text{ and } $f_t =  f^*_t \text{ otherwise}$.
    \iffalse
    \[
    f_t = \begin{cases}
        f^*_t + \epsilon & \text{if } t = t_1 \\
        f^*_t - \epsilon & \text{if } t = t_2 \\
        f^*_t & \text{otherwise}
    \end{cases}
    \]
    \fi
    Obviously $f$ is still a probability distribution on $[n]$.  Moreover, we claim that it is still a feasible solution if we choose a small enough $\epsilon$.  To see this, let $B = \{1,2, \dots, t_1-1\}$, let $M = \{t_1, t_1 + 1, \dots, t_2-1\}$, and let $A = \{t_2, t_2+1, \dots, n\}$.  For all $x \in B$, we have that $\sum_{t \in I_{\gamma_i}(x)} f_t = \sum_{t \in I_{\gamma_i}(x)} f^*_t \leq \delta_i$ for all $i \in [k]$, as required.  
    \iffalse
    \begin{align*}
        \sum_{t \in I_{\gamma_i}(x)} f_t &= \sum_{t \in I_{\gamma_i}(x)} f^*_t \leq \delta_i
    \end{align*}    
    for all $i \in [k]$, as required.  
    \fi
    
    For $x \in M$, the mass in $I_{\gamma_i(x)}$ could be larger in $f$ than in $f^*$ (since $t_1$ could be in their bad interval), but by definition of $t_2$ this was strictly less than $\delta$ in $f^*$, so by choosing a small enough $\epsilon$ we can keep it below $\delta$.  Slightly more formally, we have that $\sum_{t \in I_{\gamma_i}(x)} f_t \leq \sum_{t \in I_{\gamma_i}(x)} f^*_t + \epsilon \leq \delta$ for small enough $\epsilon$.  
    \iffalse
    \begin{align*}
        \sum_{t \in I_{\gamma_i}(x)} f_t \leq \sum_{t \in I_{\gamma_i}(x)} f^*_t + \epsilon \leq \delta
    \end{align*}
    for small enough $\epsilon$.  
    \fi
    
    For $x \in A$, note that it is impossible for $\{t_1, t_2\} \cap I_{\gamma_i}(x) = \{t_1\}$ for any $i \in [k]$; this is straightforward from the definition of $I_{\gamma_i}(x)$. So there are three cases.
    \begin{enumerate}
        \item If $\{t_1, t_2\} \cap I_{\gamma_i}(x) = \{t_1, t_2\}$, then $\sum_{t \in I_{\gamma_i}(x)} f_t = \sum_{t \in I_{\gamma_i}(x)} f^*_t \leq \delta$ since $f^*$ is feasible.
        \item If $\{t_1, t_2\} \cap I_{\gamma_i}(x) = \{t_2\}$, then $\sum_{t \in I_{\gamma_i}(x)} f_t < \sum_{t \in I_{\gamma_i}(x)} f^*_t \leq \delta$.
        \item If $\{t_1, t_2\} \cap I_{\gamma_i}(x) = \emptyset$, then $\sum_{t \in I_{\gamma_i}(x)} f_t = \sum_{t \in I_{\gamma_i}(x)} f^*_t \leq \delta$.
    \end{enumerate}
    
    Hence $f$ is feasible by Lemma~\ref{lem:bad-interval}.

    Now let's consider the competitive ratios $\alpha_f(x)$.  We break into three cases for $x$.
    \begin{enumerate}
        \item If $x \in B$, then we have that $\alpha_f(x) = \alpha_{f^*}(x) \leq \OPT$ (the final inequality is due to the optimality of $f^*$).
        \item If $x \in M$ then the competitive ratio is worse in $f$ than in $f^*$, i.e., $\alpha_f(x) > \alpha_{f^*}(x)$.  But by the definition of $M$ (and $t_2$) we know that $\alpha_{f^*}(x) < \OPT$, so by choosing a small enough $\epsilon$ we still have that $\alpha_f(x) \leq \OPT$.
        \item If $x \in A$, then it is not hard to see that the competitive ratio decreases, i.e., $\alpha_f(x) \leq \alpha_{f^*}(x) \leq \OPT$.
    \end{enumerate}

    Thus, we have a feasible solution $f$ with $\max_{x \in [n]} \alpha_f(x) \leq \OPT$, so $f$ is actually optimal.  But since $n \in A$, we know that $\alpha_f(n) < \alpha_{f^*}(n) \leq \OPT$.  This contradicts Lemma~\ref{lem:CR-n}.  Hence no such $t_1$ can exist, which implies the theorem.    
\end{proof}

\subsection{Algorithm and Analysis} \label{sec:algorithm}

With Theorem~\ref{thm:opt-tight-general} in hand, we can now give an algorithm to compute the optimal solution.  Intuitively, since Theorem~\ref{thm:opt-tight-general} says that for every day $i$ either the competitiveness constraint or the tail constraint is tight, we can just iterate through the days, increasing the probability for day $j$ until some constraint becomes tight.

We now formalize this.  We are given $n$ and $\{(\gamma_i, \delta_i)\}_{i \in [k]}$. 
As before, we will be concerned with the case where $\delta_i < 1$ (since larger values of $\delta_i$ imply the tail constraint is trivially satisfied) and $\gamma_i \geq 2 - 1/n$ for all $i$.  In what follows we will assume that we have a guess $\lambda$ which is equal to $\OPT$; we will discuss later in Appendix~\ref{sec:guessing_OPT} how to remove this assumption.  Our algorithm is the following.

\begin{itemize}
    \item Set $f_1 = \min(\min_{i \in [k]} \delta_i, \frac{\lambda-1}{n-1})$.
    \item For $j = 2$ to $n$: Set
    \begin{align*}
        f_j &= \min \left(\min_{i \in [k] : j \in I_{\gamma_i}(j)} \left( \delta_i - \sum_{t \in I_{\gamma_i}(j) \setminus \{j\}} f_t\right),  \frac{j}{n-1}  (\lambda-1) - \sum_{t=1}^{j-1} \left(1 - \frac{j-t}{n-1}\right) f_t\right) 
    \end{align*}
 \end{itemize}

We can now prove that this algorithm is optimal. The proof proceeds by showing that an optimum solution must coincide with 
$f$ to make the competitiveness or tail constraint tight, as required by Theorem~\ref{thm:opt-tight-general}. The proof is deferred to Appendix~\ref{app:omitted-general-tail}.

\begin{theorem} \label{thm:alg-optimal}
    The function $f$ returned by our algorithm is the unique optimal solution.
\end{theorem}

\section{Single Tail Constraint} \label{sec:general-gamma}

In this section we consider a single tail constraint, $(\gamma, \delta)$. To streamline our presentation, we assume $\gamma$ is an integer, but our analysis should be easily generalizable to any $\gamma \geq 2 - 1/n$.

\subsection{Analysis Overview}
We start the analysis by defining disjoint intervals $P_0, P_1, P_2, \ldots$, where the optimum purchase distribution $f_t$ exhibits a recurring structure (Section~\ref{sec:defining-P}).  
For $\gamma = 2$ which we assume for the illustration purpose here, the intervals are defined as $P_0 = [1, \frac{1}{2}(n-1)], P_1 = [ \frac{1}{2}(n-1) + 1, \frac{3}{4}(n-1)],   P_2 = [\frac{3}{4}(n-1) + 1, \frac{7}{8}(n-1)] \ldots$. Note that the intervals have exponentially decreasing lengths. 
The key observation we make in Section~\ref{sec:defining-P} is that 
every $P_j$ is a bad interval and therefore the probability mass of $f_t$ on $P_j$ is at most $\delta$ (Lemma~\ref{lem:interval-bounded-general-gamma}). 

We show how $f_t$ accumulates $\delta$ probability mass in each interval $P_j$. We first consider $P_0$ in Section~\ref{sec:p0}. Recall that $f_t$ must make the competitive constraint or the tail constraint tight at every time (Theorem~\ref{thm:opt-tight-general}). Obviously, the tail constraint doesn't get tight until $f_t$ accumulates $\delta$ mass in the beginning. The time, denoted as $\hat t$, is shown to be in the first interval $P_0$ (Lemma~\ref{lem:first-delta-hit}). Then using the fact that the competitiveness constraint must be satisfied for all $t \leq \hat t$, we can show that $f_t$ is roughly proportional to $e^{t / n}$ for all $t \leq \hat t$, then drop to 0 because the total probability mass of $f$  must be at most $\delta$ in $P_0$.

The subsequent intervals $P_j$, $j \geq 1$, are considered in Section~\ref{sec:pj}.
Here, the key observation we make is that the tail constraint is satisfied for all $t > \hat t$ (Lemma~\ref{lem:mass-tight-general-gamma}).
Thus we have a sequence of equations and by solving them, we obtain $f_t = f_{2t - (n-1) - 1} + f_{2t - (n-1)}$. Here, when $t \in P_j$, the times appearing in the right-hand-side, $2t - (n-1) - 1$ and $2t - (n-1)$ are both in $P_{j-1}$. Intuitively, this implies that $f_t$ grows twice faster in $P_j$ than in $P_{j-1}$. In fact, by a careful induction we can precisely show that $f$'s probability mass at time $t$ for $t \in P_j$ must be equal to that of some $2^j$ consecutive time steps in $P_0$ (Lemma~\ref{lem:Pj-general-gamma}). Thus, $f_t$ grows in $P_j$ with $2^j$ factor larger exponent than it does in $P_0$ (Corollary~\ref{cor:big-exponentials-general-gamma}). It then accumulates $\delta$ probability mass in $P_j$ and drops to 0 because the probability mass shouldn't exceed $\delta$ in $P_j$
(Lemma~\ref{lem:Pj-general-gamma}).

Finally, we take $n \rightarrow \infty$ in Section~\ref{sec:continuous-general-gamma} to obtain a more intuitively looking pdf in the continuous setting (Corollary~\ref{cor:continuous-general-gamma}). We note that we consider the discrete version for analysis because the recursive argument needs considerable care and it seems easier in the discrete setting. 

\subsection{Defining Disjoint Intervals, $P_j$, $j \geq 0$}

\label{sec:defining-P}

We will assume  that $n - 1$ is a sufficiently large power of $\gamma$, this will make the notation simpler since we will be interested in the optimum solution for sufficiently large values of $n$.  

Let 
\[
\ell_j =\frac{n-1}{\gamma-1}\left(1 - \frac{1}{\gamma^j}\right) + 1 \qquad \text{ and } \qquad r_j = \frac{n-1}{\gamma-1}\left(1 - \frac{1}{\gamma^{j+1}}\right)
\]
We define an interval $P_j = [\ell_j, r_j]$ for $j \geq 0$. Note that $\ell_0 = 1$, and more generally:
\begin{align}
P_0 &= \left\{t : 1 \leq t \leq \frac{n-1}{\gamma}\right\}  \nonumber \\
P_j &= \left\{t : \frac{(n-1)}{(\gamma - 1)} \cdot \frac{(\gamma^j -1)}{\gamma^j} +1  \leq t \leq \frac{(n-1)}{(\gamma - 1)}\cdot \frac{(\gamma^{j+1} - 1)}{\gamma^{j+1}} \right\} & \text{for all } j \geq 1 \label{eqn:def-P}
\end{align}  
We can check that $|P_j| = \frac{n-1}{\gamma^{j+1}}$. Hence these intervals are non-empty for all $j < \log_\gamma n$. Moreover they are disjoint. Simple algebra shows the following. 

\begin{claim}
    \label{claim:rl-general-gamma}
    For all $j \geq 1$, we have 
    \begin{itemize}
        \item $\ell_{j-1}  = \gamma \ell_j  - (n-1) - (\gamma-1)$; and 
        \item $r_{j-1}  = \gamma r_j - (n-1)$
    \end{itemize}    
\end{claim}
\begin{proof}
By definition of $r_j$, we have,
\begin{align*}
    \gamma \cdot r_j - n + 1 =  \gamma \left(\frac{n-1}{\gamma - 1} \left (1 - \frac{1}{\gamma^{j+1}}\right)\right)  - (n-1)
    & = \frac{n-1}{\gamma - 1} \left(\gamma - \frac{1}{\gamma^j}  - (\gamma - 1)\right)
    = \frac{n-1}{\gamma - 1} \left(1 - \frac{1}{\gamma^j}\right) 
     = r_{j-1}, 
\end{align*}
which proves the second claim. The first claim follows a similar algebra. 
\end{proof}

We will first show that the probability mass inside each of them is bounded. 

\begin{lemma} \label{lem:interval-bounded-general-gamma}
    For any feasible solution $f$, we have that $\sum_{t \in P_j} f_t \leq \delta$ for every $j < \log_\gamma n - 1$ 
\end{lemma}
\begin{proof}
From the definition of $I_{\gamma}(x)$, we have that
    \begin{align*}
        I_{\gamma}(r_j) &= \{ t : \gamma r_j - n + 1 < t \leq r_j \} = \{ t: r_{j-1} < t \leq r_j\} = P_j
    \end{align*}
    where the penultimate equality holds due to Claim~\ref{claim:rl-general-gamma}, and     
    the final equality is due to the strict inequality of the lower bound on $t$.  So the tail constraint  at $r_j$ implies that $\sum_{t \in P_j} f_t = \sum_{t \in I_\gamma(r_j)} f_t \leq \delta$ due to 
    Theorem~\ref{thm:opt-tight-general}, as claimed. 
\end{proof}

\subsection{Understanding the Optimum Solution's Structure for $t \in P_0$}
\label{sec:p0}

We are going to claim that for any fixed $i \geq 0$, by setting $\delta$ small enough (as a function of $i$, not $n$) the structure inside of every $P_j$ for $j \leq i$ is both simple and surprising: there is a prefix which is a $\exp(\gamma^j)$ function, and then it drops to $0$ (at least for $n$ large enough that $P_j$ is nontrivially large).  So, in other words, the optimal solution can exhibit an \emph{arbitrarily large number of drops to $0$}, and can also exhibit \emph{arbitrarily large growth}!  This is in stark contrast to the classical solution, which is a simple exponential function everywhere.

Fix $i \geq 0$, and let $\delta < 1/(2\gamma(i+1))$. Assume $n$ is sufficiently large for now; later we will require
$n \geq 2 \gamma^{i+1} +1$. 
Let $\hat t$ be the largest integer such that
\begin{equation} \label{eq:hat-t-general-gamma}
\sum_{t=1}^{\hat t} \frac{\lambda-1}{n-1} \left( 1 + \frac{1}{n-1}\right)^{t-1} \leq \delta.
\end{equation}

\begin{lemma}    
    \label{lem:first-delta-hit}
    For all sufficiently large $n$, it is the case that $\hat t < n/\gamma$, i.e., $\hat t  \in P_0$
\end{lemma}
\begin{proof}
    Since the optimum competitive ratio is at least $\frac{e}{e - 1}$ with no tail constraints, we have $\lambda \geq \frac{e}{e - 1}$. Knowing that the left-hand-side is increasing in $\hat t$, it suffices to show that 
$$
\sum_{t=1}^{n / \gamma} \frac{ e / (e - 1) -  1}{n-1} \left( 1 + \frac{1}{n-1}\right)^{t-1} > \frac{1}{2\gamma} \geq \delta.
$$
Indeed, for all sufficiently large $n$, we have, 
\begin{align*}
\sum_{t=1}^{n / \gamma} \frac{ e / (e - 1) -  1}{n-1} \left( 1 + \frac{1}{n-1}\right)^{t-1} = \frac{1}{e -1} \Big((1 + \frac{1}{n-1})^{n/\gamma} - 1\Big)
\geq \frac{e^{1/\gamma} - 1}{e - 1} \geq \frac{\nicefrac{1}{\gamma}}{e-1} >\frac{1}{2\gamma}, 
\end{align*}
which is no smaller than $\delta$ by definition.    
\end{proof}

Let $\{f_t\}_{t \in [n]}$ be the optimal solution, which we know by Theorem~\ref{thm:alg-optimal} is returned by our algorithm.  We first show that $P_0$ consists of an exponentially increasing function, followed by the zero function.
\begin{lemma} \label{lem:P0-general-gamma}
    $f_t = \frac{\lambda-1}{n-1} \left( 1 + \frac{1}{n-1}\right)^{t-1}$ for all $t \leq \hat t$, and $f_t = 0$ for all $\hat t + 2 \leq t \leq \frac{n-1}{\gamma}$.
\end{lemma}
\begin{proof}
   Let's start with the first part of the lemma, and focus on the $t \leq \hat t$ case.  For sufficiently large $n$ (as a function of $\delta$), we know from the definition of our algorithm that $f_1 = \frac{\lambda-1}{n-1}$, as claimed, which makes the competitiveness constraint tight for time $1$. 
    By Theorem~\ref{thm:opt-tight-general}, the competitiveness constraint will stay tight until there is $\delta$ mass in the bad interval for some $t$.  Lemma~\ref{lem:discrete-exponential} implies that to keep the competitiveness constraint tight, $f_t = \left(1 + \frac{1}{n-1}\right) f_{t-1}$.  Hence $f_t =  \frac{\lambda-1}{n-1} \left( 1 + \frac{1}{n-1}\right)^{t-1}$ as claimed until we have accumulated $\delta$ mass total, which by definition occurs at time $\hat t +1$, implying the first part of the lemma. 
    
   For the second part of the lemma, note that by definition of $\hat t$ the tail constraint first becomes tight at time $\hat t +1$, i.e., $\sum_{t=1}^{\hat t +1 } f_t = \delta$.  But for all $t \in [\hat t +2, n/\gamma]$, we know that $I_2(t) = [1,t]$.  Hence for all $t \in [\hat t+2, n/\gamma]$, we know that $\sum_{t' \in I_2(t) \setminus \{t\}} f_{t'} = \delta$, and hence $f_t = 0$ as claimed.
\end{proof}

\subsection{Understanding the Optimum Solution's Structure for $t \in P_j$, $j \geq 1$}
\label{sec:pj}

We now prove that the tail constraint is tight for a large range of values.
\begin{lemma} \label{lem:mass-tight-general-gamma}
    For all $t \in \cup_{j : 0 \leq j \leq i} P_j \setminus [\hat t]$, the tail constraint is tight.
\end{lemma}
\begin{proof}
    For $P_0$, this is implied by the proof of Lemma~\ref{lem:P0-general-gamma}. 
    
    Now consider $1 \leq j \leq i$, and let $t \in P_j$.  We know from Lemma~\ref{lem:interval-bounded-general-gamma} that there can be at most $\delta$ total mass inside of each $P_j$.  This allows us to bound $\alpha_f(t)$ as follows.  
    \begin{align*}
        \alpha_f(t) &= \sum_{t'=1}^t \frac{n+t'-1}{t} f_{t'} + 1 - \sum_{t'=1}^t f_{t'} = 1 + \sum_{t'=1}^t \frac{n+t'-1-t}{t} f_{t'} \\
        &= 1 + \sum_{j' = 0}^{j-1} \sum_{t' \in P_{j'}} \frac{n+t'-1-t}{t} f_{t'} + \sum_{t' \in P_j :  t' \leq t} \frac{n+t'-1-t}{t} f_{t'} 
        \tag{Disjointness of $P_0, P_1, \ldots, P_j$} 
        \\
        &\leq 1 + \sum_{j'=0}^{j-1} \frac{n+r_{j'} - 1 - t}{t} \delta + \frac{n-1}{t} \delta \tag{Lemma~\ref{lem:interval-bounded-general-gamma}} \\
        &\leq 1 + j\frac{n-1}{t} \delta + \frac{n-1}{t} \delta \tag{$r_{j'} \leq t$ for all $j'$}  \\
        &= 1 + (j+1) \frac{n-1}{t} \delta  \\
        &< 1 + (j+1) \frac{n-1}{t} \frac{1}{2\gamma(i+1)} \tag{def of $\delta$} \\
        &\leq 1 + \frac{n-1}{2\gamma t} \tag{$j \leq i$} \\
        &< 3/2 \tag{$t \geq n/\gamma$}
    \end{align*}

    Since $3/2 < \lambda$, the competitive ratio constraint cannot be tight anywhere, so the tail constraint must be tight everywhere by Theorem~\ref{thm:opt-tight-general}.
\end{proof}

Now we can analyze the solution structure of each $P_j$ for $j \leq i$.  For each such $j$, let $\hat t_j = \lfloor \hat t / \gamma^{j} \rfloor$. 

\begin{lemma} \label{lem:Pj-general-gamma}
    Let $t \in P_j$ for $1 \leq j \leq i$.
    \begin{itemize}
        \item If $t \in [\ell_j, \ell_j + \hat t_j - 1]$, then
        \begin{align*}
            f_t &= \sum_{t' = 1 + \gamma^j(t - \ell_j)}^{1 + \gamma^j(t - \ell_j) + (\gamma^j - 1)} f_{t'}
        \end{align*}
        \item If $t \in [\ell_j + \hat t_j + 1, r_j]$ then $f_t = 0$.
    \end{itemize}
\end{lemma}
\begin{proof}
    We proceed by induction on $j$. We begin with the first part of the claim. 
    Consider any $1 \leq j \leq i$. 
    Lemma~\ref{lem:mass-tight-general-gamma} implies that for every $t \in P_j$, the tail constraint is tight at both $t$ and $t-1$.  This implies that 
    \begin{equation}
    \label{eq:recursion}
        f_t = \sum_{t' \in I_\gamma(t-1) \setminus I_\gamma(t)} f_{t'} =  f_{\gamma\cdot t-(n-1) - (\gamma - 1)} + \ldots + f_{\gamma \cdot t- (n - 1)}
    \end{equation}
    First observe that if $t \in P_j$ the entries on the right hand side of Equation \ref{eq:recursion} are in $P_{j-1}$. Let $w_t = t - \ell_j$ then we have:
        \[
        \gamma (\ell_j + w_t)  - (n-1) - (\gamma -1)  = \ell_{j-1} + \gamma w_t.
    \]
    by Claim~\ref{claim:rl-general-gamma}.
    Since $|P_{j-1}| = \gamma|P_j|$, we can conclude that $\ell_{j-1} + \gamma w_t \in P_{j-1}$. A similar calculation holds for the last point in the summation. 

    We will prove the Lemma by induction on $j$. Let $j = 1$ be the base case. Then, $\ell_0 = 1$ and $\ell_1 = (n-1)/\gamma + 1$, and by Claim~\ref{claim:rl-general-gamma}, the summation in the Lemma statement begins at:
    \[\ell_{j-1} + \gamma(t - \ell_j) =  \gamma \cdot t - (n-1) - (\gamma - 1),
    \]
    which is equivalent to the first term in Equation \ref{eq:recursion}. Combined with the fact that both sums carry on for $\gamma-1$ consecutive steps, this proves the base case. 

    For the inductive step, let $j > 1$, let $t \in P_j$ and let $t = \ell_j + w_t$. Then:
    \begin{align*}
        f_t &= \sum_{t' \in I_\gamma(t-1) \setminus I_\gamma(t)} f_{t'}\\
        &=  f_{\gamma\cdot t-(n-1) - (\gamma - 1)} + \ldots + f_{\gamma \cdot t- (n - 1)}\\
        & = f_{\ell_{j-1} + \gamma w_t} + \ldots + f_{\ell_{j-1} + \gamma w_t + \gamma-1} \\
        &= \sum_{t' = 1 + \gamma^{j-1} \gamma w_t}^{1+\gamma^{j-1} \gamma w_t + \gamma^{j-1} - 1} f_{t'} + \ldots + \sum_{t' = 1 + \gamma^{j-1}(\gamma w_t + \gamma - 1)}^{1 + \gamma^{j-1}(\gamma w_t + (\gamma - 1)) + \gamma^{j-1} -1} f_{t'}\\
        & = \sum_{t' = 1 + \gamma^j w_t}^{1 + \gamma^j w_t + \gamma^{j}-1} f_{t'}.
    \end{align*}
Here the penultimate equality follows by induction and the last line follows because of the summands represent disjoint and consecutive intervals. This proves the first claim of the lemma.

        We now consider the second claim of the lemma. The base case starts from $j = 0$, and the claim holds thanks to Lemma~\ref{lem:P0-general-gamma}; recall that $r_0 = (n-1) / \gamma$.
        Now suppose that $t \in [\ell_j + \hat t_j + 1, r_j]$.  Then $w_t = t - \ell_j \geq \hat t_j +1$.  
        So $\gamma w_t \geq \gamma(\hat t_j +1) = \gamma \lfloor \hat t / \gamma^j \rfloor + \gamma \geq \lfloor \hat t / \gamma^{j-1}  \rfloor + 1 = \hat t_{j-1}+1$.  Thus, we have 
        \begin{align*}
            \gamma t - (n-1) - (\gamma - 1)         
         - (\ell_{j-1} + \hat t_{j-1} + 1) 
         &= \gamma (\ell_j + w_t)   - (n-1) - (\gamma - 1)   - (\ell_{j-1} + \hat t_{j-1} + 1)\\            &= \gamma w_t - (\hat t_{j-1} + 1) \geq 0,            
        \end{align*}
        where the last equality follows from Claim~\ref{claim:rl-general-gamma}.                  %   
         Further, we already showed above that $\gamma t - (n-1) - (\gamma-1), \ldots, \gamma t - (n-1)$ are all in $P_{j-1}$.
        Therefore, they are all  in $[\ell_{j-1} + \hat t_{j-1} + 1, r_{j-1}]$.  Thus by induction we have that $f_t 
 = f_{t\gamma - (n-1) - (\gamma-1)} + \ldots + f_{t\gamma - (n-1)}  = 0$ as claimed. \qedhere
\end{proof}

We can now combine this with the explicit formulas from Lemma~\ref{lem:P0-general-gamma} to give expressions for $f_t$ for $t \in P_j$.

\begin{lemma} \label{lem:explicit-general-gamma}
    Let $t \in P_j$.
    \begin{itemize}
        \item If $t = \ell_j$ then
        \[ f_t = \sum_{t'=1}^{\gamma^j} \frac{\lambda-1}{n-1} \left(1 + \frac{1}{n-1}\right)^{t'-1}. \]
        \item If $t \in [\ell_j+1, \ell_j + \hat t_j -1 ]$ then
        \[ f_t = \sum_{t' = \gamma^j(t-\ell_j)+1}^{\gamma^j (t-\ell_j+1) } \frac{\lambda-1}{n-1} \left(1 + \frac{1}{n-1}\right)^{t'-1}.  \]
        \item If $t \in [\ell_j + \hat t_j + 1, r_j]$ then $f_t = 0$.
    \end{itemize}
\end{lemma}
\begin{proof}
    The first and second cases follow directly from Lemmas~\ref{lem:P0-general-gamma} and~\ref{lem:Pj-general-gamma}, where we use the fact that since $t \leq \ell_j + \hat t_j -1$, we know that $\gamma^j(t - \ell_j + 1) \leq \gamma^j \hat t_j  = \gamma^j ( \lfloor \hat t / \gamma^j  \rfloor  ) \leq \hat t$, and hence the first case of Lemma~\ref{lem:Pj-general-gamma} applies to all indices in the sum.   The final case is directly from the second case of Lemma~\ref{lem:Pj-general-gamma}.
\end{proof}

This clearly implies that the beginning of $P_j$ is a $\gamma^j$ exponential:
\begin{corollary} \label{cor:big-exponentials-general-gamma}
    Let $t \in [\ell_j+1, \ell_t + \hat t_j - 1]$.  Then $f_t = \left(1 + \frac{1}{n-1}\right)^{\gamma^j} f_{t-1}$.
\end{corollary}
\begin{proof}
    By Lemma~\ref{lem:explicit-general-gamma}, we know that 
    \begin{align*}
        f_{t} &= \sum_{t' = \gamma^j(t - \ell_j)+1}^{\gamma^j(t +1 - \ell_j)} \frac{\lambda-1}{n-1} \left(1 + \frac{1}{n-1}\right)^{t'-1} \\
        &= \left(1 + \frac{1}{n-1}\right)^{\gamma^j} \sum_{t' = \gamma^j(t - \ell_j)+1}^{\gamma^j(t +1 - \ell_j)} \frac{\lambda-1}{n-1} \left(1 + \frac{1}{n-1}\right)^{t'-1 - \gamma^j} \\
        &= \left(1 + \frac{1}{n-1}\right)^{\gamma^j} \sum_{t' = \gamma^j(t-\ell_j-1)+1}^{\gamma^j(t - \ell_j)} \frac{\lambda-1}{n-1} \left(1 + \frac{1}{n-1}\right)^{t'-1} \\
        &= \left(1 + \frac{1}{n-1}\right)^{\gamma^j} f_{t-1}
    \end{align*}
    as claimed.
\end{proof}

We finally quantify how large $n$ should be compared to $i$. Note that Lemma~\ref{lem:explicit-general-gamma} is well defined if and only if the third case interval $[\ell_j + \hat t_j + 1, r_j]$ is be non-empty, i.e., $\ell_j + \hat t_j + 1 \leq r_j$. Thus, we need the following:
\begin{equation}
    \label{eqn:valid-j}
\frac{(n-1)}{(\gamma - 1)} \cdot \frac{(\gamma^j -1)}{\gamma^j} +1  
 + \lfloor \hat t / \gamma^j \rfloor +1 \leq \frac{n-1}{\gamma-1} ( 1 - 1 / \gamma^{j+1})
\end{equation}

Since $\hat t < \frac{n}{\gamma}$ (Lemma~\ref{lem:first-delta-hit}) and $n-1$ is a power of $\gamma$, we have $\hat t \leq \frac{n-1}{\gamma}$. Thus, equation~(\ref{eqn:valid-j}) is satisfied if 
\begin{equation}
    \nonumber
    \label{eqn:valid-j2}
 \frac{n-1}{\gamma - 1} \cdot \left( 1- \frac{1}{\gamma^j}\right) +1  
 + \frac{n-1}{\gamma^{j+1}} +1 \leq \frac{n-1}{\gamma-1} \left( 1 - \frac{1}{\gamma^{j+1}}\right),
\end{equation}
when $\frac{n-1}{\gamma^{j+1}} \geq 1$. By simplifying the inequality we obtain, 
$n \geq \frac{2\gamma^{j+1}}{ \gamma-1} + 1$. Both inequalities hold when $n \geq 2 \gamma^{i+1} +1$.

\subsection{Continuous Case} \label{sec:continuous-general-gamma}

Recall that in the continuous setting all times are between $0$ and $1$, and it costs $1$ to buy.  This is equal to the limit of the discrete case if we first reparameterize so all times are between $0$ and $1$ (i.e., time $t$ becomes time $t/n$) and then take the limit as $n$ goes to infinity.  That is, we can first just ``rename'' each time $t$ to $t/n$, write the resulting expressions that arise from renaming, and then take the limit as $n$ goes to infinity.  

When we do this renaming for the above results, we first get that the interval $P_j$ is equal to 

\begin{align*}
P_j &= \left\{t : \frac{1}{n} \cdot \frac{(n-1)}{(\gamma - 1)} \cdot \frac{(\gamma^j -1)}{\gamma^j} + \frac{1}{n}  \leq t \leq \frac{1}{n} \cdot \frac{(n-1)}{(\gamma - 1)}\cdot \frac{(\gamma^{j+1} - 1)}{\gamma^{j+1}} \right\} & \text{for all } j \geq 0
\end{align*}

Similarly, $\ell_j$ is defined to be $\frac{1}{n} \cdot \frac{(n-1)}{(\gamma - 1)} \cdot \frac{(\gamma^j -1)}{\gamma^j} + \frac{1}{n}$ , and we define $\hat t$ as before (except now it is a number between $0$ and $1$) and $\hat t_j$ essentially as before, to $\lfloor \hat t n / \gamma^j\rfloor / n$.  Then with this new parameterization, Corollary~\ref{cor:big-exponentials-general-gamma} turns into the following lemma.

\begin{lemma} \label{lem:big-exponentials-reparamaterized-general-gamma}
    Let $t$ be an integer multiple of $1/n$ in the interval  $[\ell_j+(1/n), \ell_t + \hat t_j - (1/n)]$.  Then $f_t = \left(1 + \frac{1}{n-1}\right)^{\gamma^j} f_{t-1}$.
\end{lemma}

A corollary of this lemma is that, if we consider two times $t' > t$ that are both integer multiples of $1/n$ in $[\ell_j+(1/n), \ell_t + \hat t_j - (1/n)]$, then $f_{t'} / f_t = \left(1 + \frac{1}{n-1}\right)^{(t'-t)n \gamma^j}$.  This is just because there are precisely $(t'-t)n$ integer multiples of $1/n$ between $t'$ and $t$.  

So in the continuous case, where we now take the limit as $n \rightarrow \infty$, we get that 
\[
\frac{f_{t'}}{f_t} = \lim_{n \rightarrow \infty} \left(1 + \frac{1}{n-1}\right)^{(t'-t)n \gamma^j} = e^{\gamma^j (t'-t)},
\]
simply because $\lim_{n \rightarrow \infty} \left(1+\frac{1}{n-1}\right)^n = e$.  Moreover, we now have a simple formula for $\hat t$, since after we reparameterize and take the limit as $n$ goes to $\infty$ in~\eqref{eq:hat-t-general-gamma} we get that $\hat t$ is the value  where $(\lambda-1) \int_0^{\hat t} e^t\, dt = \delta$, which is precisely equal to $\ln\left(1 + \frac{\delta}{\lambda-1}\right)$.  So this is our value of $\hat t$, and $\hat t_j = \hat t / \gamma^j$.  

Finally, we note that when we reparameterize and take the limit as $n$ goes to infinity of the final case of Lemma~\ref{lem:explicit-general-gamma}, we get that in the continuous setting $f(t) = 0$ for $t \in (\ell_j + \hat t_j, \ell_{j+1})$ Hence in the continuous setting we get the following corollary.

\begin{corollary} 
    \label{cor:continuous-general-gamma}
Let $\gamma \geq 2$ be an integer. 
    Let $\{P_j\}_{j \geq 0}$ be a partition of $(0, 1 / (\gamma - 1)]$, defined as $P_0 = (0, \frac{1}{\gamma})$ and  $P_j = (\frac{1}{\gamma-1} \cdot (1 - \frac{1}{\gamma^j}), \frac{1}{\gamma-1} \cdot (1 - \frac{1}{\gamma^{j+1}})]$ for all $j \geq 1$. Let $\ell_j$ be the start point of $P_j$.
    In the continuous version of the ski rental problem with a single tail constraint $(\gamma, \delta)$, the optimum solution $f(t)$ has the following structure in every $P_j$ such that $j 
 < 1 / (6\delta) - 1$:
    \begin{itemize}
        \item $f(t) = c_j \cdot e^{\gamma^j t}$ for $t$ in the interval $(\ell_j, \ell_j + \hat t_j]$, and
        \item $f(t) = 0$ for all $t \in (\ell_j + \hat t_j, \ell_{j+1}]$,
    \end{itemize}
    where $c_j$ and $\hat t_j <  \ell_{j+1} - \ell_j$ are some constants depending on $\delta$ and $j$.
    \end{corollary}
\section{Conclusion}

In this work we extended the classic ski rental problem by taking the tail risk into account.  While the problem has been studied for decades, and is perhaps the simplest online problem, there had been no previous investigation of how to balance expected performance and the risk of the randomized online algorithms performing worse than the deterministic option.  We gave a characterization theorem of the optimal purchasing distribution, and used this theorem to prove several surprising properties of the optimal distribution even under a single tail constraint.  At a high level, we showed that the optimal distribution has almost \emph{none} of the nice properties that occur without tail constraints: it is non-monotone, can alternate between regions of zero probability and non-zero probability arbitrarily many times, and can grow arbitrarily quickly even when continuous.  We also gave an explicit description of the special case of pure tail constraints, and developed an algorithm to compute the optimal purchase distribution efficiently.  

We hope that our work inspires further investigation of tail bounds in online algorithms.  As discussed, in online settings we cannot simply repeat an algorithm multiple times in order to convert a bound on the expectation into a bound on the tail, as we usually do in offline computational settings.  Yet surprisingly, there seems to be little work on tail bounds in online settings: even the ski rental problem, the most basic online setting of all, had not been investigated prior to this work!  What about other online settings where there are well-known and well-understood randomized algorithms, e.g., online matching, TCP aggregation, or the many variants of prophet inequalities?  Can we characterize the optimal randomized algorithms  for those problems in the presence of tail bounds?  Do those optimal algorithms differ significantly from the pure expectation setting, as they do for ski rental?

% Acknowledgment section to comment out for non-submission version at a later date
\section*{Acknowledgements}

We would like to thank Robert Kleinberg for fruitful discussions and comments on an earlier draft of this work.

\bibliographystyle{plain}
\bibliography{ski-rental}

\newpage
\appendix
\section{Visualising the Effect of Tail Constraints}
\label{app:figure}
We consider the problem with a single $(2, 0.31)$ tail constraint, set $n = 1000$ and plot three curves: the optimal solution, $f_t$, the expected competitive ratio $\alpha_f(t)$, and the probability of competitive ratio $\alpha_f(x)$ exceeding $\gamma = 2$. 

Visually the plot demonstrates Theorem \ref{thm:opt-tight-general} and Corollary \ref{cor:continuous-general-gamma}. Note that until the the purchase distribution falls to zero, the competitive ratio constraint is tight. This tightness is violated when the purchase distribution falls to $0$, which happens precisely when the tail risk constraint becomes tight. 

When the purchase distribution can be non-zero again (after time $\nicefrac{n}{2} = 500$), $f_t$ grows at $e^{2t}$ until the competitive ratio becomes tight, and the curve shifts to a single exponential, as described by Lemma \ref{lem:discrete-exponential}. The pattern repeats when the tail risk constraint becomes tight, with the recovery now at $e^{4t}$ at time $750$.

\begin{figure}[h!]
\centering
\includegraphics[width=0.55\textwidth]{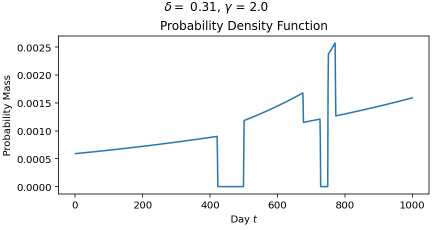}\\
(a) \\
\!\includegraphics[width=0.56\textwidth]{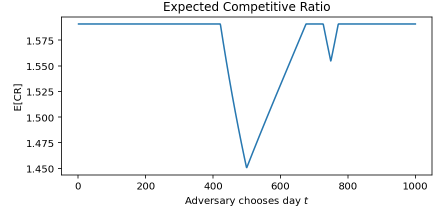}\\
(b) \\
\includegraphics[width=0.53\textwidth]{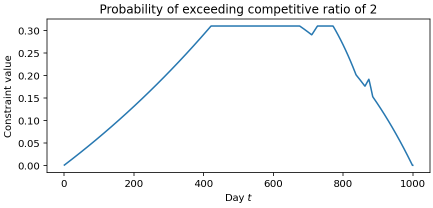}\\
(c) 
\caption{The effect of constraints of the optimal solution for $\delta = 0.31$ and $\gamma = 2$. (a) The purchase distribution, (b) the expected competitive ratio as a function of the adversary's choice, (c) the probability of exceeding $\gamma$ as a function of adversary's choice.
}
\label{fig:constraints}
\end{figure}
\section{The Randomized Algorithm Can Perform Worse Than the Deterministic Algorithm}
    \label{app:worse-than-deterministic}

For simplicity we consider the continuous version of the problem. Here, we show that the probability that the randomized algorithm that buys skis at time $t \in [0, 1]$ according to pdf $\pi$, $e^t / (e -1)$ has a worse competitive ratio with probability $1 / (\sqrt e +1)$ than the deterministic algorithm that buys skis at time 1.

As before, we reserve $x$ to denote the last time Alice skis. Let $T$ be the time Alice buys skis, sampled from the distribution $\pi$. In this case, $\opt = x$ and the algorithm pays $1 + T$ when $T \leq x$, and $x$ otherwise. Thus, we would like to know the following quantity, 
$$\max_{x \in [0, 1]} \underbrace{\Pr[ 1+ T \geq 2x \mbox{ and } T \leq x]}_{\textnormal{Let } h(x) 
   }$$

We consider two cases. 
\paragraph{Case i. $x \leq 1/2$.} In this case $h(x) = \Pr[0 \leq T \leq x] = \frac{1}{e - 1}\int_{t = 0}^{1/2} e^t dt = \frac{1}{\sqrt e +1 }$.
\paragraph{Case ii. $x \geq 1/2$.} In this case $h(x) =\Pr[2x - 1 \leq T \leq x] = \frac{1}{e-1} (e^x - e^{2x -1})$. Let $z = e^x$, then we have $z \in [e^{0.5}, e]$ and $h(x) = \frac{1}{e-1} ( z- z^2 / e)$, which is maximized when $z = e/ 2$, i.e., $x = 1/2$.

\smallskip
In both case, the maximum probability is $1 / (\sqrt e +1) \approx 0.3775$, which is achieved when $x = 1/2$.

\section{Non-Elementary Discontinuities} \label{app:lambert}
We now show that, at least in some cases, the optimal purchase distribution cannot be described using elementary functions.  We do this by showing that the location of certain discontinuities can only be described via the Lambert $W$ function.  For simplicity, we will work in the continuous setting.

Consider a single $(2,\delta)$ tail constraint, and let $\lambda = \opt$ be the optimal expected competitive ratio under this tail constraint.  For concreteness, think of $\delta$ as being something relatively large, e.g., $\delta=0.31$, in which case the optimal solution is given in Figure~\ref{fig:constraints}(a) (except that since we are in the continuous case both axes need to be rescaled). Let $f$ be this optimal solution.

It is not hard to see that the continuous version of Lemma~\ref{lem:discrete-exponential}, and particularly the continuous version of Lemma~\ref{lem:P0-general-gamma}, implies that $f$ is initially $f(t) = (\lambda-1)e^t$ (this is the same reparamaterization as from Section~\ref{sec:continuous-general-gamma}).  Note that if there is no tail constraint then we know that $\lambda = \frac{e}{e-1}$, and hence we get the function $f(t) = \frac{1}{e-1} e^t$, which is precisely the classical randomized solution.

Let $t_1$ be the time when the tail constraint first becomes tight, i.e., when $f(t)$ drops down to $0$.  This is the first time $t_1$ such that 
\[
\int_0^{t_1} (\lambda-1)e^t\, dt = \delta.
\]
If we solve for $t_1$, we get that
\begin{align*}
    & (\lambda-1) (e^{t_1} - 1) = \delta \\
    \implies & (\lambda-1) e^{t_1} = \delta+\lambda-1 \\
    \implies & e^{t_1} = \frac{\delta+\lambda-1}{\lambda-1} \\
    \implies & t_1 = \ln\left(\frac{\delta+\lambda-1}{\lambda-1}\right) = \ln\left( 1 + \frac{\delta}{\lambda-1}\right).
\end{align*}

Now we know from (the continuous version of) Theorem~\ref{thm:opt-tight-general} and the definition of $I_2(t)$ that $f(t) = 0$ for $t_1 < t \leq 1/2$ (which is also easy to see from Figure~\ref{fig:constraints}(a)).  At time $1/2 + dt$ we know from Theorem~\ref{thm:opt-tight-general} that either the tail or the competitiveness constraint must be tight, and since the competitive ratio at this point is strictly less than $\lambda$ (due to the fact that $f(t) = 0$ from $t_1$ to $1/2$), it must be the tail constraint that is tight.  An analysis similar to the base case of Lemma~\ref{lem:Pj-general-gamma} (but in the continuous setting) implies that the only way to keep the tail constraint tight is for $f(t) = 2 \cdot f(2t-1)$, which means that as long as the tail constraint remains tight, 
\[
f(t) = 2 \cdot(\lambda-1) e^{2t-1} = \frac{2(\lambda-1)}{e} e^{2t}.
\]

If the tail constraint remains tight through the interval $1/2 \leq t \leq 3/4$, then $f$ acts exactly as $P_1$ in Lemma~\ref{lem:Pj-general-gamma}: $f(t) = \frac{2(\lambda-1)}{e} e^{2t}$ until time $2t_1 - 1 = 2\ln\left(1 + \frac{\delta}{\lambda-1}\right)$, and then $f(t) = 0$ for $2t_1-1 \leq t \leq 3/4$.  However, if $\delta$ is large enough (e.g., $\delta = 0.31$), then the tail constraint \emph{will not} remain tight throughout the $[1/2, 3/4]$ interval: the competitiveness constraint will become tight before time $3/4$.  Let $t_2$ denote the time at which the competitiveness constraint becomes tight.  Now we can write an expression for $t_2$ by using (the continuous version of) Equation~\eqref{eq:alpha} and our above characterization of $f(t)$ for $t \leq t_2$, to get that $t_2$ is the first time at which
\begin{align*}
    \alpha_f(t_2) = \int_0^{t_1} \frac{1+t}{t_2} (\lambda-1) e^t\, dt + \int_{1/2}^{t_2} \frac{1+t}{t_2} \cdot \frac{2(\lambda-1)}{e} e^{2t}\, dt + 1 - \int_0^{t_1} (\lambda-1) e^t\, dt - \int_{1/2}^{t_2} \frac{2(\lambda-1)}{e} e^{2t}\, dt = \lambda
\end{align*}

Simplifying, we get that
\begin{align*}
    & \frac{1}{t_2} \int_0^{t_1} (1+t)e^t\, dt + \frac{2}{t_2 e} \int_{1/2}^{t_2} (1+t)e^{2t}\, dt - \int_{0}^{t_1} e^t\, dt - \frac{2}{e} \int_{1/2}^{t_2} e^{2t}\, dt = 1 \\
    \implies & \frac{1}{t_2} \left(t_1 e^{t_1}\right) + \frac{1}{2t_2 e} \left(\left(2t_2+1\right)e^{2t_2} - 2e\right) - \left(e^{t_1} - 1\right) - \frac{1}{e}\left(e^{2t_2}-e\right) = 1 \\
    \implies & \frac{1}{t_2}\left(t_1 e^{t_1}\right) + \frac{e^{2t_2}}{2t_2 e} -\frac{1}{t_2} - e^{t_1} = -1 \\
    \implies & \frac{1}{t_2}\left(t_1 e^{t_1}\right) + \frac{e^{2t_2}}{2t_2 e} -\frac{1}{t_2} = \frac{\delta}{\lambda-1} \\
    \implies & \frac{1}{t_2} \left( \left(1 + \frac{\delta}{\lambda-1}  \right) \ln \left( 1 + \frac{\delta}{\lambda-1} \right) + \frac{e^{2 t_2}}{2e} - 1\right) = \frac{\delta}{\lambda-1}
\end{align*}

Note the second term: $\frac{e^{2t_2}}{2t_2 e}$.  Since this does not cancel out with anything, when solving for $t_2$ we must invert this function, which gives us the Lambert W-function.  Hence $t_2$ is non-elementary. 
\section{Omitted Proofs}
    \label{app:omitted}

\subsection{From Section~\ref{sec:prelim}} \label{app:omitted-prelims}

\begin{proof}[Proof of  Lemma~\ref{lem:discrete-exponential}]
    Let's first prove the if direction.  To simplify notation, we will let $Z = \sum_{t=1}^{a-1} (n+t-1) f_t$, so $f_a = \frac{Z}{(a-1)(n-1)}$.  We proceed by induction over $x$ (or really over $x-a$).  
    \paragraph{Base case: $x=a$.}  We can use Equation~\ref{eq:alpha} to analyze the expected competitive ratios.
    \begin{align*}
        \alpha_f(a) &- \alpha_f(a-1) = \left(\frac{1}{a} \sum_{t=1}^a (n+t-1)f_t + 1 - \sum_{t=1}^a f_t\right) - \left(\frac{1}{a-1} \sum_{t=1}^{a-1}(n+t-1)f_t + 1 - \sum_{t=1}^{a-1} f_t \right) \\
        &= \left( \frac{1}{a} - \frac{1}{a-1} \right) \sum_{t=1}^{a-1} (n+t-1)f_t + \frac{1}{a} (n+a-1)f_a - f_a \\
        &= \frac{n-1}{a} f_a - \frac{1}{a(a-1)} Z \\
        &= \frac{1}{a(a-1)} Z - \frac{1}{a(a-1)} X \\
        &=0.
    \end{align*}
    Hence $\alpha_f(a) = \alpha_f(a-1)$ as desired.

    \paragraph{Inductive step.} Let $a < x \leq x'$, and suppose that that $\alpha_f(x-1) = \alpha_f(a-1)$.  Then we have that
    \begin{align*}
        \alpha_f&(x) - \alpha_f(x-1) = \left(\frac{1}{x} \sum_{t=1}^x (n+t-1)f_t + 1 - \sum_{t=1}^x f_t\right) - \left(\frac{1}{x-1} \sum_{t=1}^{x-1} (n+t-1)f_t + 1 - \sum_{t=1}^{x-1} f_t \right) \\
        &=\left(\frac{1}{x} - \frac{1}{x-1}\right) \sum_{t=1}^{x-1} (n+t-1)f_t + \frac{1}{x}(n+x-1) f_x - f_x \\
        &= \frac{n-1}{x} f_x - \frac{1}{x(x-1)} \sum_{t=1}^{x-1} (n+t-1) f_t \\
        &= \frac{n-1}{x} f_a \left(\frac{n}{n-1}\right)^{x-a} - \frac{1}{x(x-1)} \sum_{t=a}^{x-1} (n+t-1) f_a \left(\frac{n}{n-1}\right)^{t-a} - \frac{1}{x(x-1)} Z \\
        &= \frac{n-1}{x} f_a \left(\frac{n}{n-1}\right)^{x-a}  - \frac{1}{x(x-1)} f_a (n-1)\left( (x-1) \left(\frac{n}{n-1}\right)^{x-a} - a + 1 \right) - \frac{1}{x(x-1)} Z \\
        &= \frac{1}{x(a-1)} Z \left(\frac{n}{n-1}\right)^{x-a} - \frac{1}{x(x-1)(a-1)} Z \left( (x-1) \left(\frac{n}{n-1}\right)^{x-a} - a + 1 \right) - \frac{1}{x(x-1)} Z \\
        &= \frac{1}{x(a-1)} Z \left(\frac{n}{n-1}\right)^{x-a} - \frac{1}{x(a-1)} Z \left(\frac{n}{n-1}\right)^{x-a} + \frac{1}{x(x-1)} Z - \frac{1}{x(x-1)} Z \\
        &=0.
    \end{align*}
    Thus $\alpha_f(x) = \alpha_f(x-1)$, completing the induction and the proof of the if direction.  

    The only if direction now follows essentially trivially.  Suppose there were some other way of setting $f_x$ for $a \leq x \leq x'$ so that $\alpha_f(x) = \alpha_f(a-1)$ for all $a \leq x \leq x'$, and call these new values $f'$.  Let $\hat x$ be the  first time at which $f'_x \neq f_x$.  Then it is easy to see that $\alpha_{f'}(\hat x) \neq \alpha_f(\hat x) = \alpha_f(a-1)$ (where the final inequality is thanks to the ``if'' direction proved above.  This is a contradiction, and hence no such $f'$ can exist.  
\end{proof}

\subsection{From Section~\ref{sec:general-tail}} \label{app:omitted-general-tail}

\begin{proof}[Proof of Lemma~\ref{lem:bad-interval}]
    Let $f$ be a feasible distribution.  Then for all $i \in [k]$ and $x \in [n]$, we know that $\Pr_{t \sim f}[\alpha(t,x) > \gamma_i] \leq \delta_i$ since $f$ is feasible.  Hence $\sum_{t \in I_{\gamma_i}(x)} f_t \leq \delta$ by the definition of $I_{\gamma_i}(x)$ as required.  

    Now suppose that $\sum_{t \in I_{\gamma_i}(x)} f_t \leq \delta_i$ for all $i \in [k]$ and for all $x \in [n]$.  Then by the definition of $I_{\gamma_i}(x)$, we have that
    \[
    \Pr_{t \sim f}[\alpha(t,x) > \gamma_i] = \sum_{t \in I_{\gamma_i}(x)} f_t \leq \delta_i
    \]
    as required. 
\end{proof}

\begin{proof}[Proof of Lemma~\ref{lem:better_than_deterministic_soln}]
    Let $f$ be the distribution which puts mass 1 on day $n$ and mass 0 elsewhere.  This distribution satisfies all tail constraints since $\Pr_{t \sim f} [\alpha(t,x) > \gamma_i] \leq \Pr_{t \sim f} [\alpha(t,x) > 2-1/n] = 0 \leq \delta_i$.  Note that $\alpha_{f}(x) \leq 2-1/n$ for all $x \in [n]$.  We construct a new distribution $f'$ by moving $\epsilon>0$ mass from $f_n$ to $f_{n-1}$ and show that $\alpha_{f'}(x) < 2-1/n$.  Let $f'$ be the distribution given by
    \[
    f'_t =
    \begin{cases}
        1-\epsilon \quad & \text{if } t=n \\
        \epsilon \quad & \text{if } t=n-1 \\
        0 \quad & \text{otherwise}.
    \end{cases}
    \]
    To check the tail constraints we note that the only non-empty bad interval that matters corresponds to $x = n-1$ since $I_{\gamma_i}(n) = \emptyset$ for $\gamma_i \geq 2-1/n$ as we assume (Observation 
    \ref{obs:no-bad-n}).    
    Thus as long as we choose $\epsilon$ so that $\sum_{t \in I_{\gamma_i}(n-1)} f'_t = \epsilon \leq \delta_i$ for all $i\in [k]$ we have a feasible distribution.  Next, it can be shown that
    \[
    \alpha_{f'}(x) = \begin{cases}
        1+ \epsilon \quad & \text{if } x=n-1 \\
        2 + \frac{1-\epsilon}{n} \quad & \text{if } x=n \\
        1 \quad & \text{otherwise}.
    \end{cases}
    \]
    Thus if we take $\epsilon \in (0,\min_i \delta_i)$, we have $\alpha_{f'}(x) < 2-1/n$ for all $x$, which certifies that $\OPT < 2-1/n$.
\end{proof}

\begin{proof}[Proof of Lemma~\ref{lem:CR-n}]
    Since $f^*$ is optimal we know that $\alpha_{f^*}(n) \leq \max_{x \in [n]} \alpha_{f^*}(x) = \OPT$, so we just need to prove that $\alpha_{f^*}(n) \geq \opt$.  Suppose for contradiction that $\alpha_{f^*}(n) < \opt$.  Let $t_1 \in [n]$ be the smallest value of $x$ such that $f^*(x) > 0$ (note that $t_1 \leq n-1$ since otherwise $f^*$ is the classic deterministic solution, which we know is non-optimal by Lemma~\ref{lem:better_than_deterministic_soln}).  Obviously $\alpha_{f^*}(x) = 1$ for all $x < t_1$.  Let $\epsilon > 0$ be an extremely small value which we will set later.  Let $f$ be the distribution obtained from $f^*$ by moving $\epsilon$ mass from $t_1$ to $n$, i.e., 
    \[
    f_t = \begin{cases}
        f^*_t - \epsilon & \text{if } t = t_1 \\
        f^*_t + \epsilon & \text{if } t = n \\
        f^*_t & \text{otherwise}
    \end{cases} 
    \]
    Obviously $f$ is still a probability distribution over $[n]$.  Moreover, since $n \not\in I_{\gamma_i}(x)$ for all $i \in [k]$ and $x \in [n]$ due to Observation~\ref{obs:no-bad-n}, this new distribution is also feasible by Lemma~\ref{lem:bad-interval}.  

    Let $B= \{1, 2, \dots, t_1 - 1\}$ and let $M = \{t_1, t_1+1, \dots, n-1\}$.  It is easy to see that $\alpha_f(t) = \alpha_{f^*}(t) = 1 < \opt$ for all $t \in B$.  It is also easy to see that $\alpha_f(t) < \alpha_{f^*}(t) \leq \opt$ for all $t \in M$ because 
    $\alpha_f(t) - \alpha_{f^*}(t) = - \epsilon \frac{n + t_1 -1 -x}{t} < 0$.    
    %
    %\mdnote{Prove formally, but obvious}.  
    On the other hand, we have $\alpha_f(n) > \alpha_{f^*}(n)$ because 
    $\alpha_f(n) -  \alpha_{f^*}(n) = -\epsilon \frac{n + t_1 -1 - n}{n} + \epsilon \frac{n + n -1 - n}{n} > 0$. Since $\alpha_{f^*}(n) < \opt$ by assumption we can choose a small enough $\epsilon$ so that $\alpha_f(n) < \opt$.

    But now we have a feasible solution $f$ in which $\alpha_f(t) < \opt$ for all $t \in [n]$.  This contradicts the definition of $\opt$.  Thus $\alpha_{f^*}(n) = \opt$ as claimed.
\end{proof}

\begin{proof}[Proof of Theorem~\ref{thm:alg-optimal}]
    Let $f^*$ be an optimal solution.  We will prove by induction that $f_i = f^*_i$ for all $i \in [n]$, which clearly implies the theorem.  

    For the base case, it is not hard to see that $f_1 = \min(\min_{i \in [k]} \delta_i, \frac{\lambda-1}{n-1})$ is the only possible feasible value which satisfies Theorem~\ref{thm:opt-tight-general} for time $1$.  To see this, first suppose that $f^*_1 > \delta_i$ for some $i \in [k]$.  Then since $1 \in I_{\gamma_i}(1)$ we have that $\sum_{t \in I_{\gamma_i}(1)} f^*_t > \delta_i$, which implies by Lemma~\ref{lem:bad-interval} that $f^*$ is not feasible.  This is a contradiction, and hence $f^*_1 \leq \delta_i$.  Similarly, suppose that $f^*_1 > \frac{\lambda - 1}{n-1}$.  Then $\alpha_{f^*}(1) = 1-f^* + n f^*_1 = f^*(1) (n-1) + 1 > \lambda = \OPT$, which contradicts the optimality of $f^*$.  Hence $f^*_1 \leq f_1$.  On the other hand, suppose that $f^*_1 < f_1$.  Then $f^*_1 < \delta_i$ for all $i \in [k]$ and $\alpha_{f^*}(1) < \lambda$, which contradicts Theorem~\ref{thm:opt-tight-general} for $x=1$.  So $f^*_1 \geq f_1$, and thus $f^*_1 = f_1$.  
    
    For the inductive case, suppose that $f_t = f^*_t$ for all $t < j$.  Let $i \in [k]$ such that $j \in I_{\gamma_i}(j)$, and suppose that $f^*_j > \delta_i - \sum_{t \in I_{\gamma_i}(j) \setminus \{j\}} f_t$.  Then we have
    \begin{align*}
        \sum_{t \in I_{\gamma_i}(j)} f_t &= f_j +  \sum_{t \in I_{\gamma_i}(j) \setminus \{j\}} f_t >\delta_i - \sum_{t \in I_{\gamma_i}(j) \setminus \{j\}} f_t + \sum_{t \in I_{\gamma_i}(j) \setminus \{j\}} f_t = \delta_i ,
    \end{align*}
    which contradicts the feasibility of $f^*$ by Lemma~\ref{lem:bad-interval}.  So $f^*_j \leq \delta_i - \sum_{t \in I_{\gamma_i}(j) \setminus \{j\}} f_t$ for all $i \in [k]$
    
    Similarly, suppose that $f^*_j > \frac{j}{n-1} (\lambda-1) - \sum_{t=1}^{j-1} \left(1 - \frac{j-t}{n-1}\right) f_t$.  Then we have that
    \begin{align*}
        \alpha_{f^*}(j) &= \sum_{t=1}^j \frac{n+t-1}{j} f^*_t + \sum_{t=j+1}^n f^*_t \\
        &= \sum_{t=1}^j \frac{n+t-1}{j} f^*_t + 1 - \sum_{t=1}^j f^*_t \\
        &= 1 + \left( \frac{n+j-1}{j} - 1\right) f^*_j + \sum_{t=1}^{j-1} \left(\frac{n+t-1}{j} - 1\right)f^*_t \\
        &= 1 + \frac{n-1}{j} f^*_j + \sum_{t=1}^{j-1} \left(\frac{n+t-1 - j}{j} \right)f_t  \tag{induction} \\
        &> 1 + \frac{n-1}{j} \left(\frac{j}{n-1} (\lambda-1) - \sum_{t=1}^{j-1} \left(1 - \frac{j-t}{n-1}\right) f_t \right) + \sum_{t=1}^{j-1} \left(\frac{n+t-1 - j}{j}\right)f_t \\
        &= 1 + \lambda - 1 - \sum_{t=1}^{j-1} \frac{n-1 -j +t}{j} f_t + \sum_{t=1}^{j-1} \left(\frac{n+t-1 - j}{j}\right)f_t \\
        &= \lambda.
    \end{align*}
    This is a contradiction since $f^*$ is optimal, so $\alpha_{f^*}(j) \leq \OPT = \lambda$.  Hence $f^*_j \leq \frac{j}{n-1} (\lambda-1) - \sum_{t=1}^{j-1} \left(1 - \frac{j-t}{n-1}\right) f_t$.  Putting this together with the previous inequalities (for each $i \in [k]$), we get that $f^*_j \leq f_j$.

    Now suppose that $f^*_j < f_j$.  Then following the above series of inequalities but switching the $>$ to $<$ implies that $\alpha_{f^*}(j) < \lambda$.  And we also have, for every $i \in [k]$ such that $j \in I_{\gamma_i}(j)$, that
    \begin{align*}
        \sum_{t \in I_{\gamma_i}(j)} f_t &= f_j +  \sum_{t \in I_{\gamma_i}(j) \setminus \{j\}} f_t < \delta_i - \sum_{t \in I_{\gamma_i}(j) \setminus \{j\}} f_t + \sum_{t \in I_{\gamma_i}(j) \setminus \{j\}} f_t = \delta_i.
    \end{align*}
    But now we have a contradiction to Theorem~\ref{thm:opt-tight-general} for $x=j$.  Thus $f^*_j \geq f_j$.  And since we already proved that $f^*_j \leq f_j$, this implies that $f^*_j = f_j$ as claimed.      
\end{proof}

\section{Removing the Assumption on Knowing \texorpdfstring{$\opt$}{OPT}
} \label{sec:guessing_OPT}

Here we complete our algorithm by providing a binary search procedure for our guess of $\OPT$.  To do this we consider a slightly modified algorithm which doesn't explicitly try to enforce $\sum_{t=1}^n f_t = 1$.  Later, in Section~\ref{sec:relating_both_algs}, we will show how this algorithm relates to the algorithm in Section~\ref{sec:algorithm}.  Let $ALG(n,\{(\gamma_i, \delta_i)\}_{i=1}^k, \lambda')$ be the algorithm described below.  Here $\lambda'$ can be thought of as a guess of  $\OPT - 1$; we will be searching for a value of $\lambda'$ such that $\lambda' + 1 \approx \OPT$.

\begin{itemize}
    \item $f_1 = \min \left( \min_{i\in [k]} \delta_i, \frac{\lambda'}{n-1} \right)$
    \item For $x > 1$ set 
    \[
    f_x = \min \begin{cases}
        \min_{i \in [k]} \left(\delta_i - \sum_{t \in I_{\gamma_i}(x) \setminus \{x\}} f_t  \right) \\
        \frac{x}{n-1} \left( \lambda' - \sum_{t < x} \frac{n+t-x-1}{x} f_t \right)
    \end{cases}
    \]
\end{itemize}

Before providing the binary search procedure, we show some properties of the above algorithm which will be useful in analyzing correctness of the search procedure.  First, we show that the algorithm solves the following linear program.

\begin{equation*} \label{eqn:lp_lambda_prime}
\begin{array}{ll@{}ll}
      \text{Maximize} & \displaystyle \sum_{t=1}^n f_t \\
     & \displaystyle\sum_{t \leq x} \frac{n+t-x-1}{x} f_t \leq \lambda' & \quad \forall x \in [n] \\
     & \displaystyle\sum_{t \in I_{\gamma_i}(x)} f_t \leq \delta_i & \quad  \forall i \in [k], \ \forall x \in [n] \\
     & f_t \geq 0 & \quad \forall t \in [n]
\end{array}
\end{equation*}

We refer to this linear program as $LP(\lambda')$ and when the context is clear also use this to denote its optimal value.  The key to analyzing our binary search procedure is given in the following theorem which connects our algorithm to $LP(\lambda')$.  We defer the proof of this theorem to Section~\ref{sec:lp_lambda_prime_analysis}.

\begin{theorem} \label{thm:alg_solves_lp_lambda_prime}
Let $f$ be the output of $ALG(n, \{(\gamma_i, \delta_i)\}_{i=1}^k, \lambda')$, then $f$ is an optimal solution to $LP(\lambda')$.
\end{theorem}

This implies the following corollary which relates $LP(\lambda')$ back to our problem of finding the optimal distribution, allowing us to complete the search procedure.

\begin{corollary} \label{cor:alg_correctness}
    Let $f \gets ALG(n,\{(\gamma_i, \delta_i)\}_{i=1}^k, \lambda')$.  We have $\sum_t f_t \geq 1$ if and only if $\OPT \leq \lambda'+1$ 
\end{corollary} 

\begin{proof}
    For the first direction, suppose that we have $F := \sum_t f_t \geq 1$.  Then define $f' = f/F$ and note that $f'$ is a distribution.  Next we show that $f'$ satisfies all tail constraints and has expected competitive ratio at most $\lambda'+1$, implying that $\OPT \leq \lambda'+1$.

    Fix any $i\in [k]$ and trip length $x$ and consider the associated tail constraint.  We have the following:
    \[
        \sum_{t \in I_{\gamma_i}(x)} f'_t = \sum_{t \in I_{\gamma_i}(x)} f_t/F  \leq \delta_i/F \leq \delta_i 
    \]
    The first inequality follows since from Theorem~\ref{thm:alg_solves_lp_lambda_prime} and the second inequality follows since $F \geq 1$.  Thus $f'$ satisfies all tail constraints

    Now consider a trip length $x$ and the expected competitive ratio under this trip length.  We have:
    \[
    \begin{split}
            \sum_{t \leq x} \frac{n+t-1}{x} f'_t + \sum_{t > x} f'_t  & = \sum_{t \leq x} \frac{n+t-1}{x} f'_t + 1 - \sum_{t\leq x} f'_t \\
            & = \sum_{t\leq x} \frac{n+t-x-1}{x} f'_t + 1 \\
            & = \frac{1}{F}\sum_{t\leq x} \frac{n+t-x-1}{x} f_t + 1 \\
            & \leq \frac{\lambda'}{F} +1 \leq \lambda' + 1
    \end{split}
    \]
    The first equality follows since $\sum_t f'_t = 1$.  The inequalities in the last line follow from Theorem~\ref{thm:alg_solves_lp_lambda_prime} and $F \geq 1$.

    For the other direction, let $f^*$ be an optimal distribution.  In this case the expected competitive ratio of $f^*$ is $\OPT \leq \lambda' + 1$.  We show that $f^*$ is feasible for $LP(\lambda')$ and so the optimal value of $LP(\lambda')$ is at least $\sum_t f^*_t \geq 1$, which by Theorem~\ref{thm:alg_solves_lp_lambda_prime} implies that $\sum_t f_t \geq 1$.

    To start, note that the tail constraints are included in $LP(\lambda')$ and these are satisfied by $f^*$.  What remains is to analyze the other constraints.  To this end, fix $x \in [n]$, and consider the first constraint in $LP(\lambda')$.  We have the following:
    \[ \begin{split}
        \sum_{t \leq x} \frac{n+t-x-1}{x} f^*_t & = \sum_{t \leq x} \frac{n+t-1}{x} f^*_t - \sum_{t \leq x} f^*_t \\
        & = \sum_{t \leq x} \frac{n+t-1}{x} f^*_t + \sum_{t > x} f^*_t -1 \\
        & \leq \OPT - 1 \leq \lambda'
    \end{split}\]
    The second line follows since $\sum_t f^*_t = 1$.  The third line follows since $f^*$ has expected competitive ratio at most $\OPT \leq \lambda'+1$.
\end{proof}

Given the algorithm above as a subroutine, we provide the binary search procedure below, which takes in an accuracy parameter $\epsilon > 0$.  At each step we maintain an interval $[\ell, u]$ such that $\OPT - 1 \in [\ell, u]$.  Initially, we can take $\ell = \frac{e}{e-1}-1$ and $u = 1-1/n$ if all $\gamma_i \geq 2-1/n$ (otherwise we can trivially take $u = n-1$).

\begin{itemize}
    \item $\lambda' \gets \frac{\ell+u}{2}$
    \item $f \gets ALG(n,\{(\gamma_i, \delta_i)\}_{i=1}^k, \lambda')$
    \item If $\sum_t f_t < 1$, then set $\ell \gets \lambda'$ and recurse.
    \item Otherwise, we must have that $\sum_t f_t \geq 1$. Proceed with the following
    \begin{itemize}
        \item Check if $u - \ell \leq \epsilon$.  In this case output $f' = f / \sum_t f_t$
        \item Otherwise recurse with $u = \lambda'$
    \end{itemize}
\end{itemize}

The correctness of the above is guaranteed by the  following theorem.  Below we use $\Gamma$ to denote the length of the initial search range.

\begin{theorem} \label{thm:binary_search}
    The binary search procedure above queries our algorithm $O(\log(\Gamma/\epsilon))$ times and outputs a feasible solution with expected competitive ratio at most $\OPT + \epsilon$.
\end{theorem}

\begin{proof}
    The bound on the number of queries follows since the size of the initial search range is $\Gamma$ and we reduce the length by a factor of 2 after each query to $ALG$ until the length of the interval is at most $\epsilon$.  The correctness follows from Corollary~\ref{cor:alg_correctness}.
\end{proof}

\subsection{Analyzing the Algorithm via \texorpdfstring{$LP(\lambda')$}{LP(lambda)}} \label{sec:lp_lambda_prime_analysis}

Now we return to the proving Theorem~\ref{thm:alg_solves_lp_lambda_prime}, which establishes that our modified algorithm solves $LP(\lambda')$. First, it is not hard to see that our algorithm produces a feasible solution.

\begin{lemma} \label{lem:lp_lambda_prime_feas}
    Let $f \gets ALG(n,\{(\gamma_i, \delta_i)\}_{i=1}^k, \lambda')$, then $f$ is feasible for $LP(\lambda')$.
\end{lemma}
\begin{proof}
    We want to show that for all $x\in [n]$, all constraints in $LP(\lambda')$ associated with $x$ are satisfies by $f$.  We carry this out by induction on $x$.  For the base case, we have $f_1 = \min (\min_{i \in [k]} \delta_i, \frac{\lambda'}{n-1} ) \geq 0$.  Clearly, this satisfies all tail constraints for $x=1$.  For the other constraint at $x=1$, we must have $\frac{n-1}{1} f_1 \leq \lambda'$, which is also satisfied by our choice of $f_1$.

    Next consider $x>1$ and assume for induction that all constraints associated with $x' < x$ are satisfied.  The constraints $\sum_{t \leq x} \frac{n+t-x-1}{x} f_t \leq \lambda'$ and $\sum_{t \in I_{\gamma_i}(x)} f_t \leq \delta_i$ are satisfied by construction, thus it remains to show that $f_x \geq 0$.  Consider $a := \min_{i \in [k]} (\delta_i - \sum_{t \in I_{\gamma_i}(x) \setminus x} f_t)$.   For all $i \in [k]$ we have $I_{\gamma_i}(x)\setminus x \subseteq I_{\gamma_i}(x-1)$, so $\sum_{t \in I_{\gamma_i}(x) \setminus x} f_t \leq \sum_{t \in I_{\gamma}(x-1)} f_t \leq \delta_i$ by the induction hypothesis.  This implies that $a \geq 0$.  Now consider $b:=  \frac{x}{n-1}(\lambda' - \sum_{t < x} \frac{n+t-x-1}{x} f_t)$.  We have $\sum_{t < x} \frac{n+t-x-1}{x} f_t \leq \sum_{t \leq x-1} \frac{n+t-x}{x-1} f_t \leq \lambda'$ by the induction hypothesis, which implies $b \geq 0$.  Therefore $f_x = \min(a,b) \geq 0$.
\end{proof}

To complete the proof of Theorem~\ref{thm:alg_solves_lp_lambda_prime}, we need the following lemma which characterizes the optimal solution to $LP(\lambda')$.  At each $x$ at least one of the constraints must be tight in an optimal solution (similar to Theorem~\ref{thm:opt-tight-general}).

\begin{lemma} \label{lem:lp_lambda_prime_tightness}
    Let $f^*$ be an optimal solution to $LP(\lambda')$.  Then for each $x \in [n]$, at least one of the following holds:
    \begin{itemize}
        \item $\sum_{t\leq x} \frac{n+t-x-1}{x} f^*_t = \lambda'$, or
        \item there exists $i \in [k]$ such that $\sum_{t \in I_{\gamma_i}(x)} f_t = \delta_i$.
    \end{itemize}
\end{lemma}

\begin{proof}
Let $f^*$ be an optimal solution to $LP(\lambda')$ and suppose for contradiction that the condition in the lemma does not hold.  Let $B \subseteq [n]$ be the set of $x$'s where the condition does not hold.  That is for all $x \in B$ and all $i \in [k]$ we have $\sum_{t \in I_{\gamma_i}(x)} f_t < \delta_i$ and $\sum_{t\leq x} \frac{n+t-x-1}{x} f_t < \lambda'$.  We claim that $B = \{x', x'+1,\ldots, n\}$ for some $x' \in [n]$.  If this is the case then we can construct a new solution:
\[
f'_t = \begin{cases}
    f^*_{x'} + \epsilon & \quad \text{ if } t=x' \\
    f^*_t & \quad \text{ otherwise}
\end{cases}
\]
where $\epsilon >0$ is small enough so that $f'$ remains feasible.  Since $B$ is a suffix of $[n]$, increasing $f'_{x'}$ only affects the constraints given by $B$ which are all slack so there exists such an $\epsilon > 0$ which retains feasibility.  However, now we have a contradiction to the optimality of $f^*$ since $\sum_t f'_t > \sum_t f^*_t$.  To complete the proof, we just need to show that the set $B$ is of the form $\{x',x'+1,\ldots,n\}$, i.e., a suffix.

Suppose that $B$ is not a suffix of $[n]$.  Let $t_1 = \min_{x \in B} x$ and let $t_2 = \min_{x \notin B, x > t_1}x$.  Since we assume that $B$ is not a suffix, $t_2$ is well-defined. Our goal is to move $\epsilon > 0$ mass from $t_2$ to $t_1$ such that feasibility is preserved and at least one constraint becomes tight at $t_1$.  After this transformation we will have $t_1 \notin B$ and $t_2 \in B$.  Note that optimality w.r.t. $LP(\lambda')$ is also preserved since this leaves the sum unchanged.  To this end, define a new solution $f'_t$ as follows:
\[
f'_t = \begin{cases}
    f^*_{t_1} + \epsilon & \quad \text{ if } t=t_1 \\
    f^*_{t_2} - \epsilon & \quad \text{ if } t=t_2 \\
    f^*_t & \quad \text{ otherwise}
\end{cases}
\]
for some $\epsilon > 0$ to be chosen later.  It can be easily seen that $\sum_t f'_t = \sum_t f^*_t$, so $f'$ is optimal for $LP(\lambda')$ if $f^*$ is optimal.  We want to show that there exists some $\epsilon > 0$ such that $f'$ is feasible.  There are several cases to check.  First, if $x < t_1$, then the constraints at $x$ aren't affected by our new solution since $f'_t = f^*_t$ for all $t < t_1$.  So $f'$ is feasible for these constraints.  Next consider $x \in \{t_1,t_1+1,\ldots, t_2-1\}$.  By definition of $t_2$, all such $x$ are in $B$ so constraints associated with such $x$'s are slack in $f^*$.  We need to choose $\epsilon > 0$ so that
\[ \sum_{t \leq x} \frac{n+t-x-1}{x} f'_t = \sum_{t \leq x} f^*_t + \frac{n+t_1-x-1}{x} \epsilon \leq \lambda' \]
 and
\[ \sum_{t \in I_{\gamma_i}(x)} f'_t \leq \sum_{t \in I_{\gamma_i}(x)} f^*_t + \epsilon \leq \delta_i \]
for all $i \in [k]$.  Since all such $x \in B$, this is possible.
Finally consider $x \geq t_2$.  We have that
\[ \sum_{t \leq x} \frac{n+t-x-1}{x} f'_t = \sum_{t \leq x} \frac{n+t-x-1}{x} f^*_t + \epsilon \left( \frac{n+t_1-x-1}{x} - \frac{n+t_2-x-1}{x} \right) < \lambda' \]
since $t_1 < t_2$, and
\[ \sum_{t \in I_{\gamma_i}(x)} f'_t \leq \sum_{t \in I_{\gamma_i}(x)} f^*_t + \epsilon - \epsilon \leq \delta_i \]
for all $i \in [k]$.  Note that we may choose $\epsilon >0$ so that at least one constraint at $t_1$ becomes tight (since $\frac{n+t-x-1}{x} \leq \frac{n+t-t_1-1}{t_1}$ for $x \geq t_1$) and the constraints at $t_2$ become slack.  This completes the construction.  At this point, if we update $B$ to be the set of $x$'s that have slack constraints in $f'$, then either $B$ is a suffix or we can repeat the above construction until this is the case, completing the proof.
\end{proof}

Given the previous lemmas, we can complete the proof of Theorem~\ref{thm:alg_solves_lp_lambda_prime}.

\begin{proof}[Proof of Theorem~\ref{thm:alg_solves_lp_lambda_prime}]
    By Lemma~\ref{lem:lp_lambda_prime_feas}, the solution constructed by our algorithm is feasible for $LP(\lambda')$.  Let $f^*$ be an optimal solution to $LP(\lambda')$, we will show by induction that $f_x = f^*_x$ for all $x \in [n]$.

    For the base case, consider $x=1$ and note that in order to be feasible, we must have $f^*_1 \leq f_1 = \min( \min_{i \in [k]} \delta_i, \frac{\lambda'}{n-1})$.  If $f^*_1 < f_1$, then $f^*$ violates Lemma~\ref{lem:lp_lambda_prime_tightness} at $x=1$.  Thus we must have $f^*_1 = f_1$.

    For $x>1$ assume for induction that $f^*_t = f_t$ for all $t < x$.  First, we show that $f^*x \leq f_x$.  By feasibility of $f^*$ and the induction hypothesis, we have that
    \[  \lambda' \geq \sum_{t \leq x} \frac{n+t-x-1}{x} f^*_t = \sum_{t < x} \frac{n+t-x-1}{x} f_t + \frac{n-1}{x}f^*_x \]
    and 
    \[ \delta_i \geq \sum_{t \in I_{\gamma_i}(x)} f^*_t = \sum_{t \in I_{\gamma_i}(x) \setminus x} f_t + f^*_x \]
    for all $i \in [k]$, which together imply that $f^*_x \leq f_x$.  Now suppose that $f^*_x < f_x$, then by induction $f^*$ has all constraints associated with $x$ slack, violating Lemma~\ref{lem:lp_lambda_prime_tightness}.  Thus $f^*_x = f_x$, completing the proof.
\end{proof}

\subsection{Relating the Modified Algorithm Back to Section~\ref{sec:algorithm}} \label{sec:relating_both_algs}

This section relates the original algorithm we presented in Section~\ref{sec:algorithm} which assumed perfect knowledge of $\OPT$ to the modified algorithm we presented in Section~\ref{sec:guessing_OPT}.  In the following, let $f^1$ be the output of the algorithm in Section~\ref{sec:algorithm} when run with $\lambda = \OPT$ and let $f^2$ be the output of the algorithm in Section~\ref{sec:guessing_OPT} run with $\lambda' = \OPT-1$.  We have the following proposition.

\begin{claim}
If $f^1$ and $f^2$ are defined as in the paragraph above  then $f^1 = f^2$.
\end{claim}
\begin{proof}
This follows by simply noting that when $\lambda = \OPT$ and $\lambda' = \OPT-1$, then each step of both algorithms compute the same quantity.
\end{proof}

\section{Special Case of Single Tail Constraint: \texorpdfstring{$\delta = 0$}{delta=0}} \label{sec:pure}
A particularly simple setting is where there is only a single tail constraint, and it is \emph{pure}: $\delta = 0$.  In other words, we must have probability $0$ of having competitive ratio worse than $\gamma$.  In this setting we will not only be able to explicitly write the optimal solution, but will additionally be able to give $\opt$ as an explicit function of $\gamma$.  The exact expressions are given in the following two theorems.  We assume for simplicity that $\gamma - 1$ divides $n-1$, and to simplify notation we will let $\lambda = \opt$.

\begin{theorem} \label{thm:pure-opt}
    Suppose that there is a single tail constraint $(\gamma, 0)$.  Then the best competitive ratio achievable under that tail constraint is exactly
    \[
        \lambda = \opt = 1+ \frac{\gamma-1}{1 + \frac{\gamma}{n-1}  \left(n\left( \left( \frac{n}{n-1}\right)^{(n-1)\left(1 - \frac{1}{\gamma-1}\right)} - 1\right) + 1\right)} 
    \]
\end{theorem}

\begin{theorem} \label{thm:pure-solution}
    Let $f$ be the optimal purchase distribution under a single pure tail constraint $(\gamma, 0)$.  Then
    \begin{align*}
        f_t = \begin{cases} 0 & t < \frac{n-1}{\gamma-1} \\ 
        \frac{\lambda-1}{\gamma}-1 & t = \frac{n-1}{\gamma-1} \\
        \frac{\gamma(\lambda-1)}{(n-1)(\gamma-1)} \cdot \left( 1 + \frac{1}{n-1}\right)^{t - t_1 - 1} & t > \frac{n-1}{\gamma-1} 
        \end{cases}
    \end{align*}
\end{theorem}

It turns out to be easiest to first prove Theorem~\ref{thm:pure-solution} and then prove Theorem~\ref{thm:pure-opt}.

\begin{proof}[Proof of Theorem~\ref{thm:pure-solution}]
Note that in the rest of this proof, $\lambda = \opt$ but since we have not yet proved Theorem~\ref{thm:pure-opt} we will not instantiate this to any particular value; it is simply whatever the optimal expected competitive ratio is.  

Since $\delta = 0$, if $t \in I_{\gamma}(x)$ for some $x \in [n]$ then we must have $f_t = 0$.  It is easy to see that $I_{\gamma}(x)$ is nonempty if and only if $x \in I_{\gamma}(x)$, which happens if and only if $x < \frac{n-1}{\gamma-1}$.  So we know that $f_t = 0$ for all $t < \frac{n-1}{\gamma-1}$, and that $I_{\gamma}(t) = \emptyset$ for all $t \geq \frac{n-1}{\gamma-1}$.  Thus Theorem~\ref{thm:opt-tight-general} implies that $\alpha_{f^*}(t) = \lambda$ for all $t \geq \frac{n-1}{\gamma-1}$, where $f^*$ is the optimal solution.  

Let $t_1 =  \frac{n-1}{\gamma-1} $ be the first time where $I_{\gamma}(t_1) = \emptyset$, and so $\alpha_{f^*}(t_1) = \lambda$ and $f_t = 0$ for all $t < t_1$.   
 So we have
\begin{align*}
    \lambda &= \alpha_{f^*}(t_1) = \sum_{t \leq t_1} \frac{n+t-1}{t_1} f_t + 1 - \sum_{t \leq t_1} f_t 
    = \frac{n+t_1-1}{t_1} f_{t_1} + 1 - f_{t_1} 
    = \frac{n-1}{t_1} f_{t_1} + 1 = (\gamma-1) f_{t_1} + 1
\end{align*}
Rearranging, we get that $f_{t_1} = \frac{\lambda-1}{\gamma-1}$.

Since we know that $f_t = \lambda = \alpha_{f^*}(t_1)$ for all $t \geq t_1$, we can apply Lemma~\ref{lem:discrete-exponential} with $a-1 = t_1$ to obtain all of the other probabilities.  In particular, Lemma~\ref{lem:discrete-exponential} implies that 
\begin{align*}
    f_{t_1+1} &= \frac{1}{t_1 (n-1)} (n+t_1-1) f_{t_1} = \frac{\gamma-1}{(n-1)^2} \left( n + \frac{n-1}{\gamma-1} - 1\right) \frac{\lambda-1}{\gamma-1} = \frac{\gamma(\lambda-1)}{(n-1)(\gamma-1)} = \frac{\gamma}{n-1} f_{t_1} 
\end{align*}
Lemma~\ref{lem:discrete-exponential} further implies that $f_t = \left(1+\frac{1}{n-1}\right) f_{t-1}$ for all $t > t_1+1$.  So we have that for all $t \geq t_1+1$,
\begin{align*}
    f_t = \frac{\gamma(\lambda-1)}{(n-1)(\gamma-1)} \cdot \left( 1 + \frac{1}{n-1}\right)^{t - t_1 - 1}. \qedhere
\end{align*}
\end{proof}

We can now use Theorem~\ref{thm:pure-solution} to prove Theorem~\ref{thm:pure-opt}.

\begin{proof}[Proof of Theorem~\ref{thm:pure-opt}]
We can  use that fact that $f$ must be a distribution to give an exact characterization of $\lambda$ in terms of $n$ and $\gamma$.  We have that
\begin{align*}
    1&= \sum_{t=1}^n f_t =f_{t_1} + \sum_{t = t_1+1}^n f_{t} = \frac{\lambda-1}{\gamma-1} + \sum_{t = t_1+1}^n \frac{\gamma(\lambda-1)}{(n-1)(\gamma-1)}  \left( 1 + \frac{1}{n-1}\right)^{t - t_1 - 1} \\
    &= \frac{\lambda-1}{\gamma-1} + \frac{\gamma(\lambda-1)}{(n-1)(\gamma-1)} \sum_{t = t_1 + 1}^n \left( 1 + \frac{1}{n-1}\right)^{t - t_1 - 1} \\
    &= \frac{\lambda-1}{\gamma-1} + \frac{\gamma(\lambda-1)}{(n-1)(\gamma-1)} \sum_{i=0}^{n-t_1-1} \left( 1 + \frac{1}{n-1}\right)^i \\
    &= \frac{\lambda-1}{\gamma-1} + \frac{\gamma(\lambda-1)}{(n-1)(\gamma-1)} \cdot \left(n\left( \left( \frac{n}{n-1}\right)^{n-t_1-1} - 1\right) + 1\right) \\
    &= \frac{\lambda-1}{\gamma-1} + \frac{\gamma(\lambda-1)}{(n-1)(\gamma-1)} \cdot \left(n\left( \left( \frac{n}{n-1}\right)^{(n-1)\left(1 - \frac{1}{\gamma-1}\right)} - 1\right) + 1\right) 
\end{align*}

Multiplying both sides by $\gamma - 1$ gives us
\begin{align*} 
\gamma-1 &= \lambda-1 + \frac{\gamma(\lambda-1)}{n-1}  \left(n\left( \left( \frac{n}{n-1}\right)^{(n-1)\left(1 - \frac{1}{\gamma-1}\right)} - 1\right) + 1\right)  \\
&= (\lambda - 1) \left( 1 + \frac{\gamma}{n-1}  \left(n\left( \left( \frac{n}{n-1}\right)^{(n-1)\left(1 - \frac{1}{\gamma-1}\right)} - 1\right) + 1\right)\right).
\end{align*}

We can now solve for $\lambda$, giving us
\begin{align*}
\lambda &= 1+ \frac{\gamma-1}{1 + \frac{\gamma}{n-1}  \left(n\left( \left( \frac{n}{n-1}\right)^{(n-1)\left(1 - \frac{1}{\gamma-1}\right)} - 1\right) + 1\right)} 
\end{align*}
as claimed.
\end{proof}

Note that if we take the limit as $n \rightarrow \infty$ and use the fact that $\lim_{x \rightarrow \infty}(1+1/x)^x = e$, we get that $\lambda$ approaches
\begin{align*}
    1 + \frac{\gamma-1}{1 + \gamma \left(e^{1 - \frac{1}{\gamma-1}} - 1\right) } 
\end{align*}

In particular, when $\gamma = 2$ we get that $\lim_{n \rightarrow \infty} \lambda = 2$, i.e., we recover the classical deterministic bound, and when $\gamma \rightarrow \infty$ we get that $\lim_{n \rightarrow \infty} \lambda = 1 + \frac{1}{e-1} = \frac{e}{e-1}$, i.e., we recover the classical randomized bound.

\end{document}